\documentclass{article}
\usepackage{amsmath, amsfonts, amssymb, amsthm}
\usepackage{graphicx}
\usepackage{algorithm}
\usepackage[noend]{algpseudocode}
\usepackage{subfig}
\usepackage{xcolor}
\usepackage{multirow}
\usepackage{url}
\usepackage{mathtools}
\usepackage[bb=boondox]{mathalpha}
\usepackage{longtable}
\DeclareMathOperator*{\argmin}{argmin}
\title{$\epsilon$-Optimal Multi-Agent Patrol using Recurrent Strategy}
\author{Deepak Mallya, Arpita Sinha, Leena Vachhani}

\newtheorem{thm}{Theorem}
\newtheorem{lem}{Lemma}
\newtheorem{cor}{Corollary}
\newtheorem{prob}{Problem}
\newtheorem{rem}{Remark}

\begin{document}
\thanks{D. Mallya, A. Sinha, and L. Vachhani are with the Center for Systems and Control Engineering, Indian Institute of Technology Bombay, Mumbai, INDIA {\texttt \{sdeepakmallya, arpita.sinha, leena.vachhani\}@iitb.ac.in}\\ Acknowledgement: Pradyumna Paruchuri has contributed to the proofreading.}

\maketitle
\begin{abstract}
    The multi-agent patrol problem refers to repeatedly visiting different locations in an environment using multiple autonomous agents. For over two decades, researchers have studied this problem in various settings. While providing valuable insights into the problem, the works in existing literature have not commented on the nature of the optimal solutions to the problem. We, first, show that an $\epsilon$-approximate recurrent patrol strategy exists for every feasible patrol strategy. Then, we establish the existence of a recurrent patrol strategy that is $\epsilon$-optimal solution to the General Patrol Problem. The factor $\epsilon$ is proportional to the discretization constant $D$, which can be arbitrarily small and is independent of the number of patrol agents and the size of the environment. This result holds for a variety of problem formulations already studied. 
    We also provide an algorithmic approach to determine $\epsilon$-approximate recurrent patrol strategy for a patrol strategy created by any method from the literature. 
    We perform extensive simulations in graphs based on real-life environments to validate the claims made in this work.    
\end{abstract}
\section{Introduction} \label{sec_intro}
Colloquially, patrolling refers to repetitively visiting essential locations in an environment (city blocks/university campuses/military bases) for surveillance, carried out primarily as a deterrent against any intrusion by adversarial forces. Today, it is cheap and easy for adversaries of any size or form to cause damage. Thus, for any establishment to be deemed secure, it must have effective surveillance as a preventive measure against possible threats. Academia has been studying this problem for more than a couple of decades. Given the rapid advancements in robotics, vision, and communication capabilities, there are many patrolling robots available~\cite{smp-robotics,enovarobotics,goalie,rbwatcher,guardforce} that move autonomously in a given environment and detect/neutralize threats and anomalies. The study of assigning tasks/paths/routes to these robots is called the Multi-Robot Patrolling Problem (MRPP).

The existing works under the ambit of MRPP have varied in scope and the extent to which they address the task~\cite{basilico2022recent,samanta2022literature,10.1007/978-3-642-19170-1_15}. We highlight a few different aspects of the problem studied in the literature. The intruder is modelled explicitly in adversarial problem settings~\cite{agmon2011multi,SLESSLIN20191,10.1145/3555776.3577802} whereas in non-adversarial settings the intruder motion or action is ignored~\cite{10611463,nagavarapu2021generalizing}. The patrol environment can be the entire area~\cite{uav_patrol} or just the perimeter~\cite{agmon2011multi,10416787}, or the graph~\cite{10582411,10683971} induced by the road network. The patrol locations in the environment can be weighted~\cite{Pasqualetti_cooperativepatrolling,10.1007/978-3-030-66723-8_7} suggesting relative importance, or {constrained}~\cite{10.1007/978-3-319-93031-2_38,8814485,10149043} by the time interval between visits, or identical~\cite{10582411,afshani_et_al}. The patrol strategy is either deterministic~\cite{8814485,9561785} or probabilistic~\cite{SLESSLIN20191,10160923} and is either executed online~\cite{8015305,9561785}, or offline~\cite{10611463,10.1007/978-3-319-93031-2_38}. Based on the information sharing between the agents, the strategy is computed in either centralized~\cite{8814485,afshani_et_al} or decentralized~\cite{10.1145/3555776.3577802,10582411} manner. Various solutions have been proposed, such as graph theoretic~\cite{nagavarapu2021generalizing,8814485,afshani_et_al}, heuristic~\cite{emp_ana,10.1145/3555776.3577802}, task assignment~\cite{10582411,9561785}, negotiation based~\cite{10.1007/978-3-031-22953-4_2,8511293}, and learning based~\cite{10160923,jana2022deep}. Finally, both analytical~\cite{Pasqualetti_cooperativepatrolling,STRANDERS201363} and empirical~\cite{emp_ana,10.1145/3555776.3577802} results are available.

This work focuses on non-adversarial patrol of the environment represented as graphs using homogeneous robots. We now review the work available in the literature under this setting. The authors in~\cite{theo_ana} proved that the solution to the Traveling Salesman Problem is optimal for a single agent. It also established that the problem of minimizing maximum Idleness is NP-hard and commented on the trade-off between partitioned and cyclic strategies. The work in \cite{smith2011optimal} and~\cite{aksaray2015distributed} deal with temporal constraints on visits and study the problem as a finite automaton. In~\cite{gen_pat}, the analysis involved determining cyclic strategies. The authors restricted their work to outerplanar graphs. The layout was divided into tree graphs in~\cite{coop} to obtain a polynomial time solution with $8$-factor approximation. In~\cite{Pasqualetti_cooperativepatrolling}, the authors considered Weighted refresh time (Idleness), and they proposed a $2(1 + \gamma)\phi$-factor solution with non-intersecting routes where $\gamma$ and $\phi$ are edge-weight and node-weight ratios respectively. Two algorithms with $\log (n)$ and $\log (\phi)$ approximation factors were presented in~\cite{pers3} for the weighted latency (Idleness) setting. In \cite{pers1}, authors consider delaying departure from locations by introducing dwell time and providing switching conditions. Graph partitioning via clustering and simulated annealing is studied in~\cite{10.1007/978-3-319-93031-2_38}. The authors in \cite{8814485} consider weighted constraints on the latency (Idleness) and suggest a $\log (\phi)$ approximation for the same. An algorithm with an approximation factor proportional to the number of patrol agents is proposed in~\cite{10.1007/978-3-030-66723-8_7}. The graph is partitioned in~\cite{afshani_et_al} and~\cite{9837406}, and each partition is assigned an unequal number of agents to patrol. Even after a vibrant community of researchers studying the problem for a couple of decades, the nature of the optimal solutions in terms of their existence and form (cyclic/partitioned/acyclic and so on) for multi-agent settings remains unanswered. Through this work, we attempt to address this gap by providing recurrent patrol strategies for various problem formulations. We establish certain guarantees on the performance of such strategies that establish a baseline in this problem class and prove the existence of solutions.

In this work, we generalize the patrol problem to encompass various problem objectives studied in the literature. We derive a key result that an $\epsilon$-optimal recurrent patrol strategy exists for the general patrol problem. The significant contributions are as follows-
\begin{itemize}
    \item We show the existence of an $\epsilon$-approximate recurrent patrol strategy for every patrol strategy. This result enables a fair comparison between different solution approaches proposed in the literature under a single class of recurrent patrol strategies.
    \item The approximation factor $\epsilon$ is proportional to the discretization constant $D$ that can be arbitrarily small.
    \item We, then, establish that an $\epsilon$-optimal recurrent patrol strategy exists for a broad class of problems studied under the non-adversarial multi-agent patrol problem. Also, this bound is an improvement over results in the existing literature.
    \item We provide a systematic approach to construct $\epsilon$-approximate recurrent patrol strategy for any patrol strategy that solves the general patrol problem.
\end{itemize}

The rest of the paper is organized as follows. Section~\ref{sec_prob} introduces the terminology used and formulates the general patrol problem. In Section~\ref{sec_rec}, we state and prove that (i) for every patrol strategy, there exists an $\epsilon$-approximate recurrent patrol strategy, and (ii) Problem~\ref{prob_gpp} defined in Section~\ref{sec_prob} has an $\epsilon$ optimal recurrent patrol strategy. We also present an algorithmic approach to convert a given strategy to an $\epsilon$-approximate recurrent strategy. Section~\ref{sec_sim} contains simulation results that validate the claims made in Section~\ref{sec_rec}. We conclude in Section~\ref{sec_con} by summarizing the essential contributions of this work and suggesting future directions for expansion.

\section{Problem Definition}\label{sec_prob}
The patrol environment is a strongly connected directed graph $\mathcal{G}(\mathcal{V}, \mathcal{E})$ where $\mathcal{V} = \{v_1, \ldots, v_{\vert \mathcal{V} \vert}\}$ is the set of nodes corresponding to the patrol locations, $\mathcal{E}$ is the set of edges representing the road segments between the locations. Each edge $(v_a, v_b) \in \mathcal{E}$ has an associated edge weight $w_{ab} \in \mathbb{R}_+$ denoting the nominal time required to traverse from node $v_a$ to node $v_b$. Each node $v_a \in \mathcal{V}$ has a weight $\phi_a \in [0, 1]$ signifying its importance. We assume a set of homogeneous patrol agents $\mathcal{A} = \{a_1, \ldots, a_{\vert \mathcal{A} \vert}\}$ that patrol the nodes by travelling along the edges.

We define patrol strategy $\pi = \{(t(i), r(i), v(i), a(i))\}_{i = 1}^\infty$ as an ordered sequence of 4-tuples. For $i = 1, 2, \dots$, agent $a(i) \in \mathcal{A}$ departs from node $v(i) \in \mathcal{V}$ at time $t(i) \in \mathbb{R}$ after staying at $v(i)$ for a duration of $r(i) \in \mathbb{R}$. Let $t^a(i) \in \mathbb{R}$ be the time of arrival at $v(i)$, then , $t(i) = t^a(i) + r(i)$. Each 4-tuple in $\pi$ corresponds to the instance when an agent leaves a node. We maintain a chronological sequence of departures in $\pi$. That is, $t(i) \geq t(i - 1)$ with $t(1) = 0$. We also assume that $r(1) = 0$. 

Consider a scenario where agent $a_k$ departs from node $v_a$ at the $m^{th}$ instance and next visits node $v_b$ at the $n^{th}$ instance after traversing along the edge $(v_a, v_b)$. We have,
\begin{align}
    t(n) &= t^a(n) + r(n) \nonumber \\
    &= t(m) + w_{ab} + r(n) \label{eq_tij}
\end{align}
where the indices $m$ and $n$ need not be consecutive.

The Instantaneous Idleness (or Idleness) value $I_k(t)$ of node $v_k$ at time $t$ is the time elapsed since an agent's most recent departure from node $v_k$. At instance $i$, as agent $a(i)$ stays at node $v(i) = v_a$ from $t^a(i)$ to $t(i)$, $I_a(t) = 0$ in the interval $[t^a(i), t(i)]$. We assume, without loss of generality, the initial Idleness value $I_k(0) = 0$ for every node $v_k \in \mathcal{V}$, under any patrol strategy $\pi$. Furthermore, we assume a constraint on the Idleness values $I_k(t)$ for a subset of the nodes $\mathcal{V}^T \subseteq \mathcal{V}$ as $I_k \leq T_k$. The ordered set of Idleness values for all nodes in $\mathcal{V}$ is referred to as the Idleness vector, denoted as $\mathcal{I}_\pi(t)$.

We propose a generalization that encapsulates several patrolling problem formulations studied in the literature. 
\begin{prob}[General Patrol Problem]\label{prob_gpp}
Given a graph $\mathcal{G}(\mathcal{V}, \mathcal{E})$, a subset $\mathcal{V}^T \subseteq \mathcal{V}$ with the set of constraints $\{T_k \mid v_k \in \mathcal{V}^T\}$, and an Idleness function $f_\pi(t, \mathcal{V})$, find a patrol strategy $\pi_\star$ that minimizes the cost
\begin{equation}
    J_\pi = \max_{t \geq 0} f_\pi(t, \mathcal{V})\label{eq_J}
\end{equation}
while satisfying $I_k(t) \leq T_k ~ \forall v_k \in \mathcal{V}^T, ~ t \geq 0$. 
\end{prob}
The Idleness function $f_\pi(t, \mathcal{V})$ can represent -
\begin{itemize}
    \item Graph Maximum Idleness~\cite{theo_ana,coop,gen_pat} \[f_\pi(t,\mathcal V) = \max_{v_k \in \mathcal{V}} I_k(t)\]
    \item Weighted Maximum Idleness~\cite{pers3,10.1007/978-3-030-66723-8_7,Pasqualetti_cooperativepatrolling,8814485} \[f_\pi(t,\mathcal V) = \max_{v_k \in \mathcal{V}} \phi_k I_k(t)\]
    \item Graph Average Idleness~\cite{emp_ana,bayes}\[f_\pi(t,\mathcal V) = \sum_{v_k \in \mathcal{V}} \frac{I_k(t)}{\vert \mathcal{V} \vert}\] 
\end{itemize}
In general, the Idleness function is as follows, 
\begin{equation}
    f_\pi(t, \mathcal{V}) = C \lVert \Phi \mathcal{I}_\pi(t) \rVert_p \label{eq_f}
\end{equation}
where $\Phi$ is a diagonal matrix consisting of node weights, $C$ is a constant, and $\lVert \cdot \rVert_p$ is the $\ell^p$-norm. The Idleness functions such as $\max_{v_k \in \mathcal{V}} I_k(t) = \lVert \mathcal{I}_\pi(t) \rVert_\infty$ corresponds to $\ell^\infty$-norm and $\sum_{v_k \in \mathcal{V}} I_k(t)/\vert \mathcal{V} \vert = C\lVert \mathcal{I}_\pi(t) \rVert_1$ corresponds to $\ell^1$-norm with some constant $C$.

We show the generalization (Problem~\ref{prob_gpp}) has a recurrent solution $\pi_\star^R$ that is $\epsilon$-optimal, that is, $J_{\pi_\star^R} \leq (1 + \epsilon(D)) J_{\pi_\star}$ with $I_k(t) \leq (1 + \epsilon(D))T_k ~ \forall v_k \in \mathcal{V}^T, ~ t \geq 0$ and some small $D \in \mathbb{R}_+$.

\section{Recurrent Patrol Strategy} \label{sec_rec}
This section shows the existence of $\epsilon$-optimal recurrent solution $\pi_\star^R$ to Problem~\ref{prob_gpp}. We arrive at the result in three steps. For any given patrol strategy $\pi$, we obtain a discrete patrol strategy $\pi^D$ and establish that $\pi^D$ is an $\epsilon$-approximate solution to $\pi$, that is, the cost of $\pi^D$, $J_{\pi^D} \leq (1 + \epsilon) J_\pi$. Then, we construct a recurrent patrol strategy $\pi^R$ from a part of $\pi^D$ and show that $\pi^R$ is also an $\epsilon$-approximate solution to $\pi$. Finally, we show that the set of recurrent strategies is finite and thereby establish the existence of an $\epsilon$-optimal recurrent solution $\pi_\star^R$.

\subsection{$\epsilon$-approximate discrete patrol strategy $\pi^D$}
Suppose there exists a patrol strategy $\pi$ that solves Problem~\ref{prob_gpp} with a cost of $J_\pi$. Consider a discretization constant $D \in \mathbb{R}_+$. We construct a new patrol strategy $\pi^{D} = \{(\tau(i), \rho(i), v(i), a(i))\}_{i = 1}^\infty$ from $\pi$. For each instance $i$ in $\pi^D$, the agent $a(i)$ stays at node $v(i)$ for $\rho(i)$ duration and departs at $\tau(i)$ where $\rho(i)$ and $\tau(i)$ satisfy,
\begin{subequations} \label{eq_step1}
\begin{gather}
    \rho(i) \geq r(i) \label{eq_r1}\\
    \tau(i) \bmod{D} = 0 \label{eq_r2}
\end{gather} 
\end{subequations}
Let $d(i)$ for $i = 1, 2, \ldots$ be the time difference in departures under $\pi$ and $\pi^D$, that is, 
\begin{equation}
d(i) = \tau(i) - t(i) \label{eq_d}    
\end{equation}
Let $\bar{d}(n)$ be the maximum over all values of $d(\cdot)$ until the index $n$. That is,
\begin{align}
    \bar{d}(n) &= \max_{i = 1, \ldots, n} d(i) \nonumber \\
    &= \max \left(d(n), \bar{d}(n - 1)\right) \label{eq_dmax}
\end{align}
Note that $d(i)$ is independent of the values of $d(\cdot)$ for $i = 1, \ldots, i - 1$. Hence, $\bar{d}(\cdot)$ need not increase at every instance.

In addition to~\eqref{eq_step1}, we impose the following two constraints on $\pi^D$ to ensure that it is $\epsilon$-approximate to $\pi$. First, the sequence of departures under $\pi^D$ is identical to that under $\pi$. Second, for each visit instance $n$, $d(n)$ is at most $D$ less than $\bar{d}(n)$. That is,
\begin{subequations} \label{eq_step2}
    \begin{align}
        \tau(n) &\geq \tau(n - 1) \label{eq_s1}\\
        d(n) &> \bar{d}(n) - D \label{eq_s2}
    \end{align}
\end{subequations}
For an agent visiting nodes $v(m)$ and then $v(n)$ under $\pi^D$,~\eqref{eq_tij} gets modified to,
\begin{align}
    \tau(n) &= \tau^a(n) + \rho(n) \nonumber \\
    &= \tau(m) + w_{ab} + \rho(n) \label{eq_tauij}
\end{align} 
where $\tau^a(n)$ is the time of arrival at instance $n$.

We set $\rho(1) = 0, \tau(1) = 0, d(1) = 0$ and $\bar{d}(1) = 0$. 
The values of $d(n), \bar{d}(n), \tau(n), \rho(n)$ for each instance $n = 2, 3, \ldots$ is determined iteratively as follows, 
\begin{subequations} \label{eq_update}
    \begin{align}
        d(n) &= \max \left(d^\prime, d^\prime + D \left\lfloor \frac{\bar{d}(n - 1) - d^\prime}{D} \right\rfloor \right) \label{eq_dn}\\
        \bar{d}(n) &= \max \left(d(n), \bar{d}(n - 1) \right) \label{eq_dbarn}\\
        \tau(n) &= t(n) + d(n) \label{eq_taun}\\
        \rho(n) &= d(n) - d(m) + r(n) \label{eq_rhon}
    \end{align}
\end{subequations}
where,
\begin{subequations} \label{eq_help}
\begin{align}
        d^\prime &= \tau^{\prime \prime} - t(n) \\
        \tau^{\prime \prime} &= \max (\tau(n - 1), \tau^\prime) \label{eq_taupp}\\
        \tau^\prime &= \tau(m) + D \left\lceil \frac{w_{ab} + r(n)}{D} \right\rceil \label{eq_taup}
\end{align}
\end{subequations}

The index $m$ in the above equations corresponds to the last departure of agent $a(n)$ before arriving at $v(n)$.

\begin{lem}
    The discrete patrol strategy $\pi^D$ constructed from $\pi$ iteratively using~\eqref{eq_update} satisfies~\eqref{eq_step1} and~\eqref{eq_step2}. \label{lem1}
\end{lem}
\begin{proof}
We prove this by induction. Equations \eqref{eq_step1} and~\eqref{eq_step2} are trivially satisfied for $n = 1$.
We show that if the constraints~\eqref{eq_step1} and~\eqref{eq_step2} are satisfied for $n - 1$, then they are also satisfied for $n$. 

We start by showing $\tau^\prime$ and $\tau^{\prime \prime}$ are multiples of $D$.
Since $m < n$, $\eqref{eq_r2}$ is satisfied at $m$. 
From~\eqref{eq_taup},
\begin{align*}
    \tau^\prime \bmod{D} &= \left(\tau(m) + D \left\lceil \frac{w_{ab} + r(n)}{D} \right\rceil \right) \bmod{D} \nonumber \\
    &= \left(\tau(m) \bmod{D}\right) + \left(D \left\lceil \frac{w_{ab} + r(n)}{D} \right\rceil \bmod{D}\right) \nonumber \\
    &= 0 
\end{align*}
In~\eqref{eq_taupp}, both $\tau(n- 1)$ and $\tau^\prime$ are multiples of $D$ making $\tau^{\prime\prime} \bmod{D} = 0$.

From~\eqref{eq_dn}, there are two possible values for $d(n)$. 

\emph{Case A:} Consider \[d^\prime \geq d^\prime + D \left\lfloor \frac{\bar{d}(n - 1) - d^\prime}{D} \right\rfloor\] 
Implying $d(n) = d^\prime$.
From~\eqref{eq_taun} and~\eqref{eq_help},
\begin{align}
    \tau(n) &= t(n) + d(n) \nonumber\\
    &= t(n) + d^\prime \nonumber \\
    &= t(n) + \left(\tau^{\prime \prime}  - t(n)\right) \nonumber \\
    &= \tau^{\prime \prime} \label{eq_ttpp}\\
    &= \max \left(\tau(n - 1), \tau^\prime \right) \nonumber \\
    &\geq \tau^\prime \nonumber \\
    &= \tau(m) +  D \left\lceil \frac{w_{ab} + r(n)}{D} \right\rceil \nonumber
\end{align}
Substituting for $\tau(n)$ from~\eqref{eq_tauij},
\begin{align}
    \tau(m) + w_{ab} + \rho(n) &\geq \tau(m) +  D \left\lceil \frac{w_{ab} + r(n)}{D} \right\rceil \nonumber \\
    \implies \rho(n) &= D \left\lceil \frac{w_{ab} + r(n)}{D} \right\rceil - w_{ab} \nonumber \\
    &\geq \left(w_{ab} + r(n)\right) - w_{ab} \nonumber \\
    \implies \rho(n) &\geq r(n) \nonumber
\end{align}
Hence,~\eqref{eq_r1} is satisfied. Since $\tau^{\prime \prime} \bmod{D} = 0$,  from~\eqref{eq_ttpp},~\eqref{eq_r2} is satisfied. From~\eqref{eq_taupp}, $\tau^{\prime \prime} \geq \tau(n - 1)$ implying ~\eqref{eq_s1} is satisfied. On substituting $d(n) = d^\prime$ in~\eqref{eq_dn} we have,
\begin{align*}
    d(n) &\geq d(n) + D \left\lfloor \frac{\bar{d}(n - 1) - d(n)}{D} \right\rfloor \\
    \implies 0 &\geq D \left\lfloor \frac{\bar{d}(n - 1) - d(n)}{D} \right\rfloor \\
    \implies d(n) &\geq \bar{d}(n - 1) \\
    &> \max \left(d(n) - D, \bar{d}(n - 1) - D\right) \\
    &= \max \left(d(n), \bar{d}(n - 1)\right) - D \\
    &= \bar{d}(n) - D
\end{align*}
Hence,~\eqref{eq_s2} is also satisfied. So far we established that constraints~\eqref{eq_step1} and~\eqref{eq_step2} are satisfied when \[d^\prime \geq d^\prime + D \left\lfloor \frac{\bar{d}(n - 1) - d^\prime}{D} \right\rfloor\] 

\emph{Case B:} Consider \[d^\prime < d^\prime + D \left\lfloor \frac{\bar{d}(n - 1) - d^\prime}{D} \right\rfloor\] 
From~\eqref{eq_dn},
\begin{equation*}
    d(n) = d^\prime + D \left\lfloor \frac{\bar{d}(n - 1) - d^\prime}{D} \right\rfloor
\end{equation*}
The term $D \left\lfloor \frac{\bar{d}(n - 1) - d^\prime}{D} \right\rfloor$ results in agent $a(n)$ staying for a longer duration at $v(n)$ than the case when $d(n) = d^\prime$. Also, this additional duration is a multiple of $D$. As constraints~\eqref{eq_r1},~\eqref{eq_r2} and~\eqref{eq_s1} are met if $d(n) = d^\prime$, they remain satisfied on adding the term $D \left\lfloor \frac{\bar{d}(n - 1) - d^\prime}{D} \right\rfloor$. To show that~\eqref{eq_s2} holds,
\begin{align*}
    d(n) &= d^\prime + D \left\lfloor \frac{\bar{d}(n - 1) - d^\prime}{D} \right\rfloor \\
    &> d^\prime + (\bar{d}(n - 1) - d^\prime) - D\\
    &= \bar{d}(n - 1) - D\\
    \implies d(n) &> \max \left(d(n) - D, \bar{d}(n - 1) - D \right) \\
    &= \max \left(d(n), \bar{d}(n - 1)\right) - D \\
    &= \bar{d}(n) - D
\end{align*}
Hence, constraints~\eqref{eq_step1} and~\eqref{eq_step2} are satisfied also when \[d^\prime < d^\prime + D \left\lfloor \frac{\bar{d}(n - 1) - d(n)}{D} \right\rfloor \qed \] 
\end{proof}
Next, we establish a relationship between $\bar{d}(\cdot)$ and $t(\cdot)$ for two different instances. For that, consider instances $p, i$ and $q$ such that the following two conditions hold. First, the agent $a(i)$ departs $v(p)$ and arrives at $v(i)$ after traversing the edge $(v(p), v(i))$ with edge-weight $\omega$ (say). From~\eqref{eq_d},
\begin{align}
    d(i) &= \tau(i) - t(i) \nonumber\\
    &= \left(\tau(p) + \omega + \rho(i)\right) - \left(t(p) + \omega + r(i)\right) \nonumber \\
    &= d(p) + \rho(i) - r(i) \label{eq_assume1}
\end{align}
Secondly, since $\bar{d}(\cdot)$ need not increase at every instance, we assume that it increases at instance $i$ and then again at instance $q + 1$.  
That is, 
\begin{equation}
    \bar{d}(i - 1) < \bar{d}(i) = \cdots = \bar{d}(q) < \bar{d}(q + 1) \label{eq_assume2}
\end{equation}

\begin{lem}
Given an instance $q$ satisfying~\eqref{eq_assume2}, there exists a unique instance $p$ such that
\begin{subequations}
\begin{align}
    \bar{d}(q) &\leq \bar{d}(p) + D \label{eq_dbar} \\
    t^a(q) - t(p) &\geq \underline{w} \label{eq_tbar}
\end{align}
\end{subequations}
where $\underline{w}$ is the minimum edge weight.
\label{lem_pq}
\end{lem}
\begin{proof}
From~\eqref{eq_dmax} and~\eqref{eq_assume2},
\[\bar{d}(q) = \bar{d}(i) = \max \left(d(i), \bar{d}(i - 1) \right) > \bar{d}(i - 1)\] 
Implying, $\bar{d}(q) = d(i)$.
From~\eqref{eq_update} and~\eqref{eq_help},
\begin{align*}
    \bar{d}(q) &= d(i) \\
    &= \max \left( d^\prime, d^\prime + D \left\lfloor \frac{\bar{d}(i - 1) - d^\prime}{D} \right\rfloor \right) \\
    &\leq \max (d^\prime, d^\prime + (\bar{d}(i - 1) - d^\prime)) \\
    &= \max (d^\prime, \bar{d}(i - 1)) \\
    &= d^\prime \\
    &= \tau^{\prime \prime} - t(i) \\
    &= \max \left(\tau(i - 1), \tau^\prime\right) - t(i) \\
    &= \max \left(\tau(i - 1), \tau(p) + D \left\lceil \frac{\omega + r(i)}{D} \right\rceil\right) -t(i) \\   
    &= \max \left(\tau(i - 1), D \left\lceil \frac{t(p) + d(p) + \omega + r(i)}{D} \right\rceil\right) -t(i) \\
    &= \max \left(\tau(i - 1), D \left\lceil \frac{t(i) + d(p)}{D} \right\rceil\right) -t(i) \\
    &\leq \max \left(\tau(i - 1) - t(i - 1), D \left\lceil \frac{t(i) + d(p)}{D} \right\rceil - t(i)\right) \\
    &\leq \max (\bar{d}(i - 1), (t(i) + d(p) + D) - t(i)) \\
    &= d(p) + D \\
    &\leq \bar{d}(p) +  D 
\end{align*}
Hence,~\eqref{eq_dbar} is proved.
Under $\pi$, agent $a(i)$ departs $v(p)$ at $t(p)$ and arrives $v(i)$ at $t^a(i)$. We have, \((t(i) - t(p)) \geq (t^a(i) - t(p)) \geq \underline{w}\). Substituting $t^a(q) \geq t(i)$, \[t^a(q) - t(p) \geq \underline{w} \qed \] 
\end{proof}
Next, we relate the Idleness values of nodes under $\pi$ and $\pi^D$. 
Under $\pi^D$, consider a time $\tau$ and a node $v_h \in \mathcal{V}$. Let $\tau(m_h) \leq \tau$ be the time of the most recent departure from $v_h$. Let $\tau^a(n) \geq \tau$ be the first instance of arrival to any node after $\tau$. The corresponding times under $\pi$ are $t(m_h)$ and $t^a(n)$.

\begin{lem}
The Idleness value $I_h^D(\tau)$ of node $v_h$ at time $\tau$ under $\pi^D$ is related to the Idleness value $I_h(t^a(n))$ at time $t^a(n)$ under $\pi$ as
\begin{equation}
    I_h^D(\tau) \leq \left(1 + \epsilon(D)\right)I_h(t^a(n)) \label{eq_ID}
\end{equation}
where $\epsilon(D)$ is some linear function of the discretization constant $D$. \label{lem3}
\end{lem}
\begin{proof}
If an agent visits $v_h$ in $\left[\tau(m_h), \tau^a(n)\right)$ under $\pi^D$, $I_h(\tau) = 0$ and therefore~\eqref{eq_ID} holds trivially. So, we prove the lemma when no agent visits $v_h$ in $\left[\tau(m_h), \tau^a(n)\right)$. The Idleness of $v_h$ at $\tau$ is, 
\begin{align}
    I_h^D(\tau) &= \tau - \tau(m_h) \nonumber\\
    &\leq \tau^a(n) - \tau(m_h) \nonumber\\
    &\leq (t^a(n) + d(n)) - (t(m_h) + d(m_h)) \nonumber \\
    &= (d(n) - d(m_h)) + (t^a(n) - t(m_h)) \nonumber \\
    \implies I_h^D(\tau) &\leq \Delta I_h + I_h(t^a(n)) \label{eq_fin}
\end{align}
where,\( \Delta I_h = d(n) - d(m_h) \leq \bar{d}(n) - d(m_h) \).
Since $\pi^D$ satisfies Lemma~\ref{lem1}, we substitute~\eqref{eq_s2} for $d(m_h)$ to obtain,
\begin{align}
   \Delta I_h &\leq \bar{d}(n) - \bar{d}(m_h) + D \label{eq_DI}
\end{align}
From Lemma~\ref{lem_pq}, there exists a unique instance $p_1$ for instance $n$ such that~\eqref{eq_dbar} and~\eqref{eq_tbar} hold. Similarly, a unique instance $p_2$ for instance $p_1$ and so on.
Suppose we have $L$ such instances $p_1, \ldots, p_L$ between instances $m_h$ and $n$,
\begin{align*}
    &\bar{d}(n) \leq \bar{d}(p_1) + D &\text{ \&}& \quad t^a(n) - t(p_1) \geq \underline{w} \\
    &\bar{d}(p_1) \leq \bar{d}(p_{2}) + D &\text{ \&}& \quad t(p_{1}) - t(p_{2}) \geq \underline{w} \\
    & \qquad \vdots & & \qquad \vdots \\
    &\bar{d}(p_L) \leq \bar{d}(m_h) + D &\text{ \&}& \quad t(p_L) \geq t(m_h)
\end{align*}
Adding the above $(L + 1)$ inequalities,
\begin{align}
    &\bar{d}(n) \leq \bar{d}(m_h) + LD + D &\text{ \&}& \quad t^a(n) \geq t(m_h) + L \underline{w} \label{eq_sum} 
\end{align}
Substituting the first inequality in~\eqref{eq_sum},
\begin{align*}
    \Delta I_h &\leq \bar{d}(m_h) + (L + 1)D - \bar{d}(m_h) + D =(L + 2)D 
\end{align*}
From the second inequality,
\begin{equation*}
    L \leq \frac{t^a(n) - t(m_h)}{\underline{w}} = \frac{I_h(t^a(n))}{\underline{w}}
\end{equation*}
Substituting for $L$ in the previous inequality,
\begin{equation*}
    \Delta I_h \leq \left( \frac{I_h(t^a(n))}{\underline{w}} + 2 \right) D
\end{equation*}
Substituting in~\eqref{eq_fin},
\begin{equation}
    I_h^D(\tau) \leq \left[1 + \left(\frac{1}{\underline{w}} + \frac{2}{I_h(t^a(n))}\right) D \right]I_h(t^a(n)) \label{I^D} 
\end{equation}  
Let $\underline{I}(\pi)$ denote the minimum non-zero Idleness value of any node over all arrival instances $t^a(i), i = 1, 2, \ldots$, that is, \[\underline{I}(\pi) = \min {\{I_j(t^a(i)) \vert I_j(t^a(i)) \neq 0, v_j \in \mathcal{V}, i = 1, 2, \ldots \}}\]
Substituting in~\eqref{I^D},
\begin{equation*}
    I_h^D(\tau) \leq \left(1 + \epsilon(D)\right)I_h(t^a(n))
\end{equation*}
where \[\epsilon(D) = \left(\frac{1}{\underline{w}} + \frac{2}{\underline{I}(\pi)}\right) D \qed\] 
\end{proof}
Now, we present the main result of this subsection.
\begin{thm}
The discrete patrol strategy $\pi^D$ constructed iteratively as per~\eqref{eq_update} is an $\epsilon$-approximate strategy to the given patrol strategy $\pi$, that is, 
\begin{equation}
J_{\pi^D} \leq \left(1 + \epsilon(D) \right) J_\pi \label{J_piD}
\end{equation}\label{thm_dps}
\end{thm}
\begin{proof}
From~\eqref{eq_J} and~\eqref{eq_f},
\begin{align*}
    J_{\pi^D} = \max_{\tau \geq 0} \lVert \Phi \mathcal{I}_{\pi^D}(\tau) \rVert_p
\end{align*}
From~\eqref{eq_ID}, for any $\tau \geq 0$,
\begin{align*}
    \lVert \Phi \mathcal{I}_{\pi^D}(\tau) \rVert_p &\leq (1 + \epsilon(D)) \lVert \Phi \mathcal{I}_{\pi}(t^a(n)) \rVert_p \\
    &\leq \left(1 + \epsilon(D)\right) \max_{i = 1, 2, \ldots} \lVert \Phi \mathcal{I}_{\pi}(t^a(n)) \rVert_p
\end{align*}
As the above inequality holds for each $\tau \geq 0$,
\begin{align*}
    J_{\pi^D} &\leq \left(1 + \epsilon(D) \right) \max_{i = 1, 2, \ldots} \lVert \Phi \mathcal{I}_\pi(t^a(i)) \rVert_p \\
    &\leq \left(1 + \epsilon(D) \right) \max_{t \geq 0} \lVert \Phi \mathcal{I}_\pi(t) \rVert_p \\
    &= \left(1 + \epsilon(D) \right) J_\pi \qed
\end{align*}
\end{proof}
\begin{cor}
Under $\pi^D$, for all nodes $v_k \in \mathcal{V}^T$ and for all $\tau \geq 0 $, 
\begin{equation}
    I_k^D(\tau) \leq \left(1 + \epsilon(D) \right) T_k \label{T_piD}
\end{equation}\label{cor_dps} 
\end{cor}
\begin{proof}
For each $\tau$, there exists a $t^a(n)$ such that~\eqref{eq_ID} holds, implying for each $v_k \in \mathcal{V}^T$,
\begin{align*}
    \max_{\tau \geq 0} I_k^D(\tau) &\leq (1 + \epsilon(D)) \max_{i = 1, 2, \ldots} I_k(t^a(i))\\
    &\leq (1 + \epsilon(D)) T_k     
\end{align*} 
The second inequality holds since $\pi$ satisfies Problem~\ref{prob_gpp}.\qed
\end{proof}
Theorem~\ref{thm_dps} and Corollary~\ref{cor_dps} proves that the Problem~\ref{prob_gpp} has a discrete solution with a $(1 + \epsilon(D))$-approximation factor.
\begin{figure}
    \centering
    \fbox{\includegraphics[width = \textwidth]{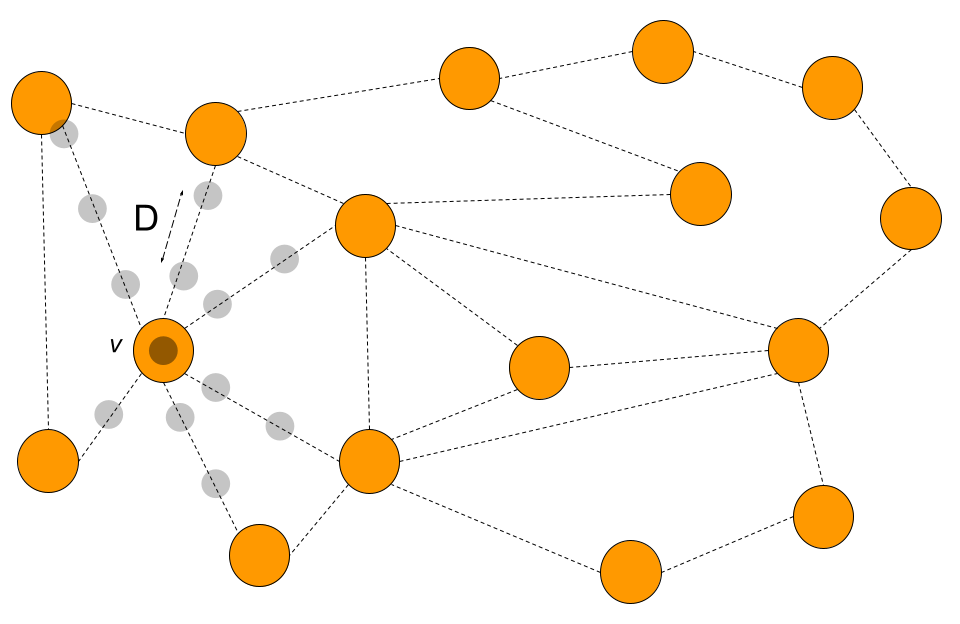}}
    \caption{A graphical representation of feasible positions for agent $a_k$ after departing from node $v$ under patrol strategy $\pi^D$.}
    \label{agent_pos}
\end{figure}
\subsection{$\epsilon$-approximate recurrent patrol strategy $\pi^R$}
Next, we construct a recurrent patrol strategy $\pi^R$ from $\pi^D$. Under $\pi^R$, each agent repeats its walk after a specific time interval. At any time $t$, let $\mathcal{X}_{\pi}(t)$ represent the set of positions of all the agents at $t$. 
\begin{lem}
    The Idleness vector $\mathcal{I}_{\pi^D}(\tau(i))$ and the set of agent positions $\mathcal{X}_{\pi^D}(\tau(i))$ belong to finite sets for all $i = 1, 2, \ldots$.\label{thm_2} 
\end{lem}
\begin{proof}
Since $\tau(i), i = 1, 2, \ldots$ is a multiple of $D$, the Idleness value $I^D_k(\tau(i))$ for any node $v_k$ is also a multiple of $D$. From~\eqref{eq_ID}, $I^D_k(\tau(i))$ is bounded, implying the Idleness vector $\mathcal{I}_{\pi^D}(\tau(i))$ is also bounded. Hence, $\mathcal{I}_{\pi^D}(\tau(i))$ belongs to a finite set.

Since we assume the speed of the agents to be the same and constant, each agent can reach only a finite set of points on any edge at every $\tau(i), i = 1, 2, \ldots$. Therefore, the set of feasible locations for an agent at time $\tau(i)$ is finite (Figure~\ref{agent_pos}). Thus, the collection of all agent positions $\mathcal{X}_{\pi^D}(\tau(i))$ also belongs to a finite set. \qed
\end{proof}

Since the Idleness vectors and set of agent positions are from finite sets, we can find at least two instances $\tau(p), \tau(q)$ with $p < q$ such that 
\begin{enumerate}
    \item[C1]$\mathcal{I}_{\pi^D}(\tau(p)) = \mathcal{I}_{\pi^D}(\tau(q))$
    \item[C2]$\mathcal{X}_{\pi^D}(\tau(p)) = \mathcal{X}_{\pi^D}(\tau(q))$
\end{enumerate}

Let $\pi^D_{pq}$ be the portion of patrol strategy $\pi^D$ from $\tau(p)$ to $\tau(q)$. As the Idleness value of each node is identical at $\tau(p)$ and $\tau(q)$, it implies all nodes are visited at least once in the interval. This portion is used to construct the recurrent patrol strategy $\pi^R$.

As $\mathcal{X}_{\pi^D}(\tau(p)) = \mathcal{X}_{\pi^D}(\tau(q))$, every agent $a_g$ has a corresponding agent $a_h$ such that the position of $a_g$ at $\tau(p)$ is identical to that of $a_h$ at $\tau(q)$. Let $\chi: \mathcal{A} \rightarrow \mathcal{A}$ be the bijective map defined as $a_h = \chi(a_g)$. Let $\chi^m(a_k) = \chi(\ldots(\chi(a_k)))$ signify $\chi$ acted on an agent $a_k$ iteratively $m$ times. 

We construct an infinite patrol strategy $\pi^R$ from $\pi^D_{pq}$ where each agent $a_k$ after traversing its path along $\pi^D_{pq}$ continues with the path of agent $\chi(a_k)$. The patrol strategy $\pi^R = \{(\theta(i), \sigma(i), \nu(i), \alpha(i)\}_{i = 1}^\infty$ is expressed in terms of $\pi^D_{pq} = \{(\tau(i), \rho(i), v(i), a(i)\}_{i = p}^q$ where,
\begin{subequations}\label{eq_rec}
\begin{align}
    &\begin{multlined} \theta(i) = (\tau(i \bmod{(q - p)} + p) - \tau(p))\\ + \lfloor i / (q - p) \rfloor (\tau(q) - \tau(p)) \end{multlined}\\
    &\sigma(i) = \rho(i \bmod{(q - p)} + p) \\
    &\nu(i) = v(i \bmod{(q - p)} + p)\\
    &\alpha(i) = \chi^{\lfloor i / (q - p) \rfloor}(a(i \bmod{(q - p)} + p))
\end{align}
\end{subequations}
Let the duration of the recurrent segment be denoted by $L(\pi^R) = \tau(q) - \tau(p)$ and the bijective map be denoted by $\chi_{\pi^R}$.
\begin{thm}
For every patrol strategy $\pi$, there exists an $\epsilon$-approximate recurrent patrol strategy $\pi^R$.\label{thm_rps}
\end{thm}
\begin{proof}
Let $I^R_h(\tau)$ and $I^D_h(\tau)$ be the Idleness value of node $v_h$ at time $\tau$ under $\pi^R$ and $\pi^D$ respectively. By construction, $I^R_h(\tau) \leq I^D_h((\tau \bmod{\tau(q) - \tau(p)}) + \tau(p)), \forall \tau \in [0, \infty)$. Hence, $f_{\pi^R}(\tau, \mathcal{V}) \leq f_{\pi^D}(\tau, \mathcal{V})$, implying, $J_{\pi^R} \leq J_{\pi^D}$. From Theorem~\ref{thm_dps}, $\pi^D$ is an $\epsilon$-approximate solution to $\pi$. Hence, $\pi^R$ is also an $\epsilon$-approximate solution to $\pi$. \qed
\end{proof}

Next, we move from $\epsilon$-approximate to $\epsilon$-optimal strategy.

\subsection{$\epsilon$-optimal recurrent patrol strategy $\pi_\star^R$}
So far, we have established that given a patrol strategy $\pi$ with a cost of $J_\pi$ there exists a recurrent patrol strategy $\pi^R$ with a cost of $J_{\pi^R} \leq (1 + \epsilon(D))J_\pi$. We now generalize the recurrent patrol strategy solving Problem~\ref{prob_gpp} and show that there exists an $\epsilon$-optimal recurrent patrol strategy $\pi^R_\star$ to Problem~\ref{prob_gpp}.

Let $\Pi^R_D$ be the collection of all the recurrent patrol strategies with discretization constant $D$, which solves Problem~\ref{prob_gpp}. These strategies are characterized by a recurrent segment that repeats infinitely many times, and the departure instances are multiples of $D$. The need not be obtained through the above procedure. We state the results in this section for the set $\pi^R_D$.
Let $\Pi^R_D(J) \subseteq \Pi^R_D$ such that for each $\pi^R \in \Pi^R_D(J)$, $J_{\pi^R} \leq J$. 
\begin{lem}
    The set $\Pi^R_D(J)$ is finite for any $J \in \mathbb{R}$. \label{thm_fin}
\end{lem}
\begin{proof}
For each $\pi^R \in \Pi_D^R(J)$, $J_{\pi^R} \leq J$, implying, the Idleness value $I_h(\tau)$ is bounded for each node $v_h$ and for all $\tau$. Therefore, from Lemma~\ref{thm_2}, the Idleness vector $\mathcal{I}_{\pi^R}(\tau(i))$ and the agent positions $\mathcal{X}_{\pi^R}(\tau(i))$ belong to a finite set for all $i = 1, 2, \ldots$. 

As $\pi^R$ is a recurrent strategy, it consists of a recurrent segment of finite length $L(\pi^R)$, i.e., $L(\pi^R) < M < \infty$, for some constant $M$. 
We have (1) $\mathcal{I}_{\pi^R}(\tau) = \mathcal{I}_{\pi^R}(\tau + L(\pi^R))$ and (2) $\mathcal{X}_{\pi^R}(\tau) = \mathcal{X}_{\pi^R}(\tau + L(\pi^R))$ for all $\tau \geq L(\pi^R)$\footnote{For $\tau < L(\pi^R)$, $\mathcal{I}_{\pi^R}(\tau) \leq \mathcal{I}_{\pi^R}(\tau + L(\pi^R))$ as Idleness values are set to $0$ at the beginning.}. 

Consider the interval $\left[L(\pi^R), 2L(\pi^R)\right)$. Let us define a sequence of 2-tuples \[\mathcal{S}_{\pi^R} = \{(\mathcal{I}_{\pi^R}(jD), \mathcal{X}_{\pi^R}(jD)) \: \vert \: L(\pi^R) \leq jD < 2L(\pi^R), \: j \in \mathbb{N}\}\] consisting of Idleness vectors and agent positions at each multiple of $D$ within the interval.
The patrol strategy $\pi^R$ can be uniquely identified with the sequence $\mathcal{S}_{\pi^R}$. 
The sequence $\mathcal{S}_{\pi^R}$ is characterized by the length $L(\pi^R)$, the sequence of Idleness vectors $\mathcal{I}_{\pi^R}(\cdot)$ and the sequence of agent positions $\mathcal{X}_{\pi^R}(\cdot)$.
We show that each of them can take finitely many values.

For a given patrol strategy $\pi^R$, the length of the recurrent segment $L(\pi^R)$ is a multiple of $D$ and bounded by $M$. Hence, $L(\pi^R)$ belongs to a finite set. Each element in the sequence of Idleness vectors $\mathcal{I}_\pi(\cdot)$ belongs to a finite set, and the length of the sequence is finite, implying the sequence of Idleness vectors $\mathcal{I}_{\pi^R}(\cdot)$ belongs to a finite set. The same holds for the sequence of agent positions $\mathcal{X}_{\pi^R}(\cdot)$. Hence the set $\{\mathcal{S}_{\pi^R} \vert \pi^R \in \Pi^R_D(J)\}$ is finite, implying $\Pi^R_D(J)$ is a finite set. \qed
\end{proof}

Now, we present the main result of this paper.
\begin{thm}
    Assume that the Problem~\ref{prob_gpp} has a solution. There always exists a recurrent patrol strategy $\pi^R_\star$ for this problem which is $\epsilon$-optimal.\label{thm_opt}
\end{thm}
\begin{proof}
From Theorem~\ref{thm_rps}, for every patrol strategy $\pi$ there exists an $\epsilon$-approximate recurrent patrol strategy $\pi^R \in \Pi^R_D$. From Lemma~\ref{thm_fin}, $\Pi^R_D(J)$ is finite for a given $J$. We consider some $J$ such that $\Pi^R_D(J)$ is non-empty. Then, \[\pi_\star^R := \argmin_{\pi^R \in \Pi^R_D(J)}{J_{\pi^R}}\] exists and
the cost of every $\epsilon$-approximate $\pi^R$ is at least that of $\pi^R_\star$. Therefore, $\pi_\star^R$ is $\epsilon$-optimal recurrent patrol strategy to Problem~\ref{prob_gpp}. \qed
\end{proof}
\begin{rem}
The reader can appreciate the significance of this result from the following example. \cite{afshani_et_al} gives a $2(1 - 1/k)$ factor approximation to the min-max latency problem where $k$ is the number of agents. We improve this factor to almost $D/\underline{w}$ where $D$ is a parameter (discretization constant), and $\underline{w}$ is the minimum edge weight of the graph. Hence, $\epsilon$ is independent of the number of patrol agents, and $\epsilon$ is arbitrarily small by selecting $D$ appropriately. Also, our result holds for a broader class of problems as defined in the General Patrol Problem.
\end{rem}

\subsection{Algorithm to generate $\epsilon$-approximate recurrent patrol strategy.}
We define two algorithms to generate a $\epsilon$-approximate recurrent patrol strategy for a given patrol strategy solving Problem~\ref{prob_gpp}. The first algorithm discretizes the given strategy described in Algorithm~\ref{alg_disc}. The second algorithm determines a recurring segment in the discretized strategy obtained from Algorithm~\ref{alg_disc} and returns a recurrent patrol strategy as given in Algorithm~\ref{alg_rec}.

\begin{algorithm}
    \caption{Discrete patrol strategy}
    \label{alg_disc}
    \emph{Input}-$D, {\pi} = \{(t(i), r(i), v(i), a(i))\}_{i = 1}^T$
    \emph{Output}-$\pi^D$
    \begin{algorithmic}[1]
        \State $\tau(1) \leftarrow 0, \: \rho(1) \leftarrow 0, \: \bar{d}(1) \leftarrow 0, \: d(1) \leftarrow 0$
        \State ${\pi}^D \leftarrow \{(\tau(1), \rho(1), v(1), a(1))\}$
        \For {$(t(i), r(i), v(i), a(i)) ~ i = 2, \dots, T$}
        \State{$m \leftarrow Last Instance(a(i))$}
        \State{$\tau^\prime \leftarrow \tau(m) + D \left\lceil \frac{w_{ab} + r(i)}{D} \right\rceil$}
        \State{$\tau^{\prime \prime} \leftarrow \max \left(\tau(i - 1), \tau^\prime \right)$}
        \State{$d^\prime \leftarrow \tau^{\prime \prime} - t(i)$}
        \State{$d(i) \leftarrow d^\prime + D \lfloor \max {(\bar{d}(i - 1) - d^\prime, 0)}/D \rfloor$}
        \State {$\tau(i) \leftarrow t(i) + d(i)$}
        \State{$\rho(i) \leftarrow d(i) - d(m) + r(i)$}
        \State {$\bar{d}(i) \leftarrow \max \left( d(i), \bar{d}(i - 1) \right)$}
        \State ${\pi}^D \leftarrow {\pi}^D \cup \{(\tau(i), \rho(i), v(i), a(i))\}$ 
        \EndFor
        \State \Return ${\pi}^D$
    \end{algorithmic}
\end{algorithm}
In Algorithm~\ref{alg_disc}, the inputs are the discretization constant $D$ and patrol strategy ${\pi} = \{(t(i), r(i), v(i), a(i))\}_{i = 1}^T$. It returns a discrete patrol strategy $\pi^D$. The for-loop (Lines $3-12$) iterates over each instance $i$ in ${\pi}$  using~\eqref{eq_update} and~\eqref{eq_help} and ensures that $\pi^D$ satisfies~\eqref{eq_step1} and~\eqref{eq_step2}.
We run the loop over a finite number of departures $T$. We select $T$ arbitrarily but large enough to ensure Algorithm~\ref{alg_rec} returns a solution.

\begin{algorithm}
    \caption{Recurrent patrol strategy} \label{alg_rec}
    \emph{Input}-$\{\mathcal{I}_{\pi^D}(\tau(i))\}_{i = 1}^T, \{\mathcal{X}_{\pi^D}(\tau(i))\}_{i = 1}^T$
    \emph{Output}-$(p, q)$
    \begin{algorithmic}[1]
        \For{$q = 2, \ldots, T, \textbf{ and } p = 1, \ldots, q - 1$}
        \If{$\mathcal{I}_{\pi^D}(\tau(p)) = \mathcal{I}_{\pi^D}(\tau(q)), ~ \mathcal{X}_{\pi^D}(\tau(p)) = \mathcal{X}_{\pi^D}(\tau(q))$} 
        \State{\Return {$(p, q)$}}
        \EndIf
        \EndFor
        \Return $Fail$
    \end{algorithmic}
\end{algorithm}

We can generate the Idleness vectors $\mathcal{I}_{\pi^D}(\tau(i)$ and agent positions $\mathcal{X}_{\pi^D}(\tau(i))$ for $i = 1, \ldots, T$ for the discrete patrol strategy $\pi^D$ obtained in Algorithm~\ref{alg_disc}.
Algorithm~\ref{alg_rec} takes these Idleness vectors and agent positions as inputs. It searches for time instances $p$ and $q$ satisfying C1 and C2 in Section~\ref{sec_rec} as given in Line $2$ and returns them as the output. If the output fails, we increase $T$ and rerun Algorithm~\ref{alg_disc} and then~\ref{alg_rec} and repeat the process until we get a pair of time instances. We, then, determine the mapping $\chi(\cdot)$ and the recurrent patrol strategy $\pi^R$ from~\eqref{eq_rec}.

\section{Simulation Results} \label{sec_sim}
In this section, we validate Thereom~\ref{thm_rps} using extensive simulations of patrol strategies in real-life environments. 

\begin{algorithm}[ht]
    \caption{Greedy-Random Patrol} \label{alg_grps}
    \emph{Input}- Time $t$, Agent $a$, Node $v_i$, Idleness vector $\mathcal{I}$, Constants $\gamma, \lambda$
    \emph{Output}- Go to Node $v_j$, Departure Instance $\tau$
    \begin{algorithmic}[1]
        \State $v_j \leftarrow \max_{v_k \in \mathcal{N}(v_i)} I_k(t)$ 
        \State $\tau \leftarrow w_{ij}$
        \If{$x \sim U([0, 1]) \leq \gamma$} $\tau \leftarrow (1 + y \sim exp(\lambda))$
        \EndIf
        \State {\Return $(v_j, \tau)$}
    \end{algorithmic}
\end{algorithm}

Greedy-Random Patrol (Algorithm~\ref{alg_grps}) generates the patrol strategy that acts as input to Algorithm~\ref{alg_disc}. It is an online algorithm; the agents are assigned nodes on the go. When an agent $a_h$ reaches node $v_i$, it is assigned the neighbour $v_j$ with maximum Idleness value (Line $1$). The departure time from $v_i$ is delayed at arbitrary instances (Line $3$) to incorporate randomness in the patrol strategy. We set the probability of such cases to $\gamma = 0.0001$. The amount of delay is an exponential random variable with $\lambda = 1$, giving a realistic patrol scenario of agents delayed due to traffic and other circumstances.

\begin{figure}[ht]
\fbox{\parbox{\textwidth}{
    \centering
    \subfloat[Campus A]{\includegraphics[width=0.32\textwidth]{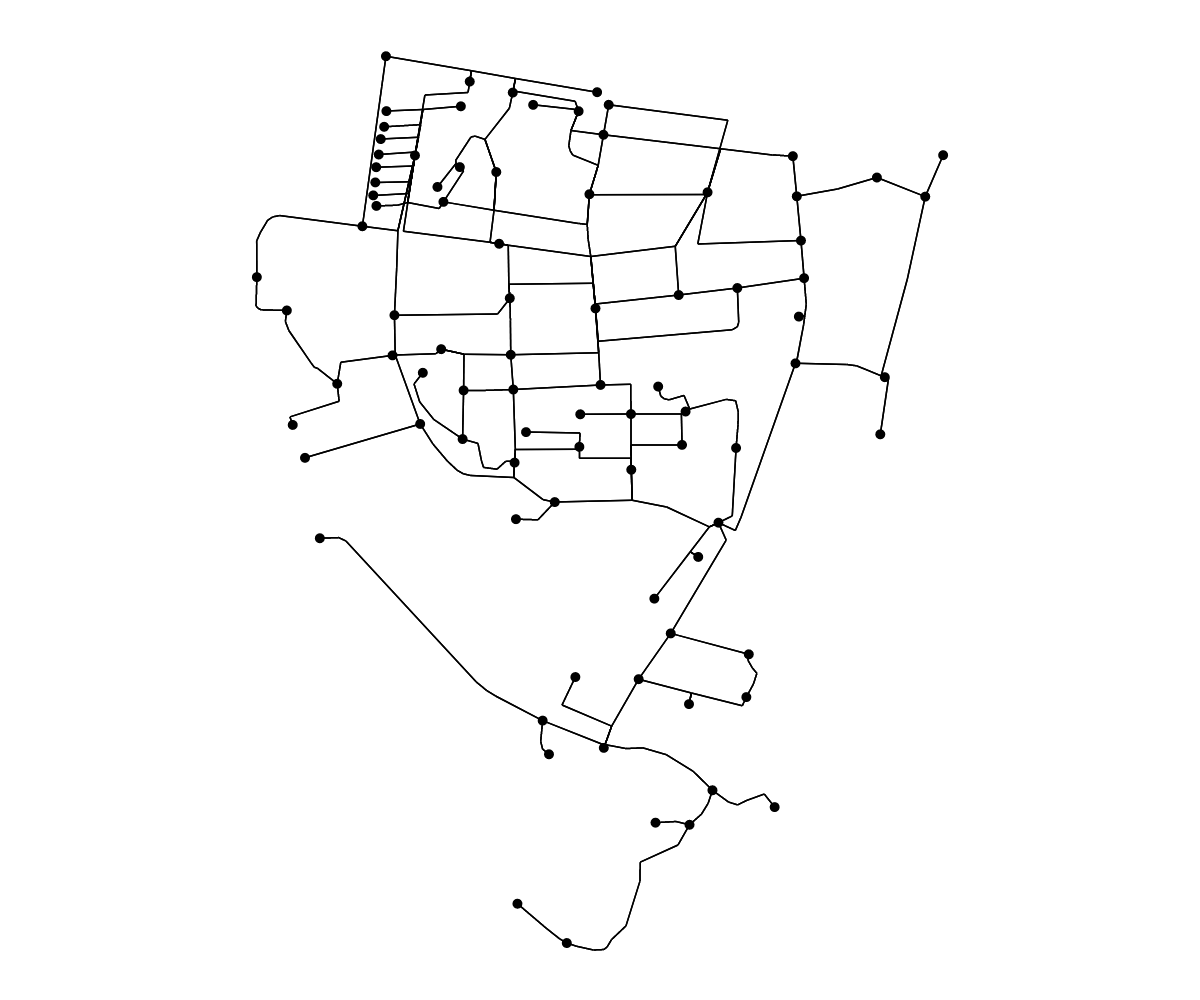}}
    \subfloat[Campus B]{\includegraphics[width=0.32\textwidth]{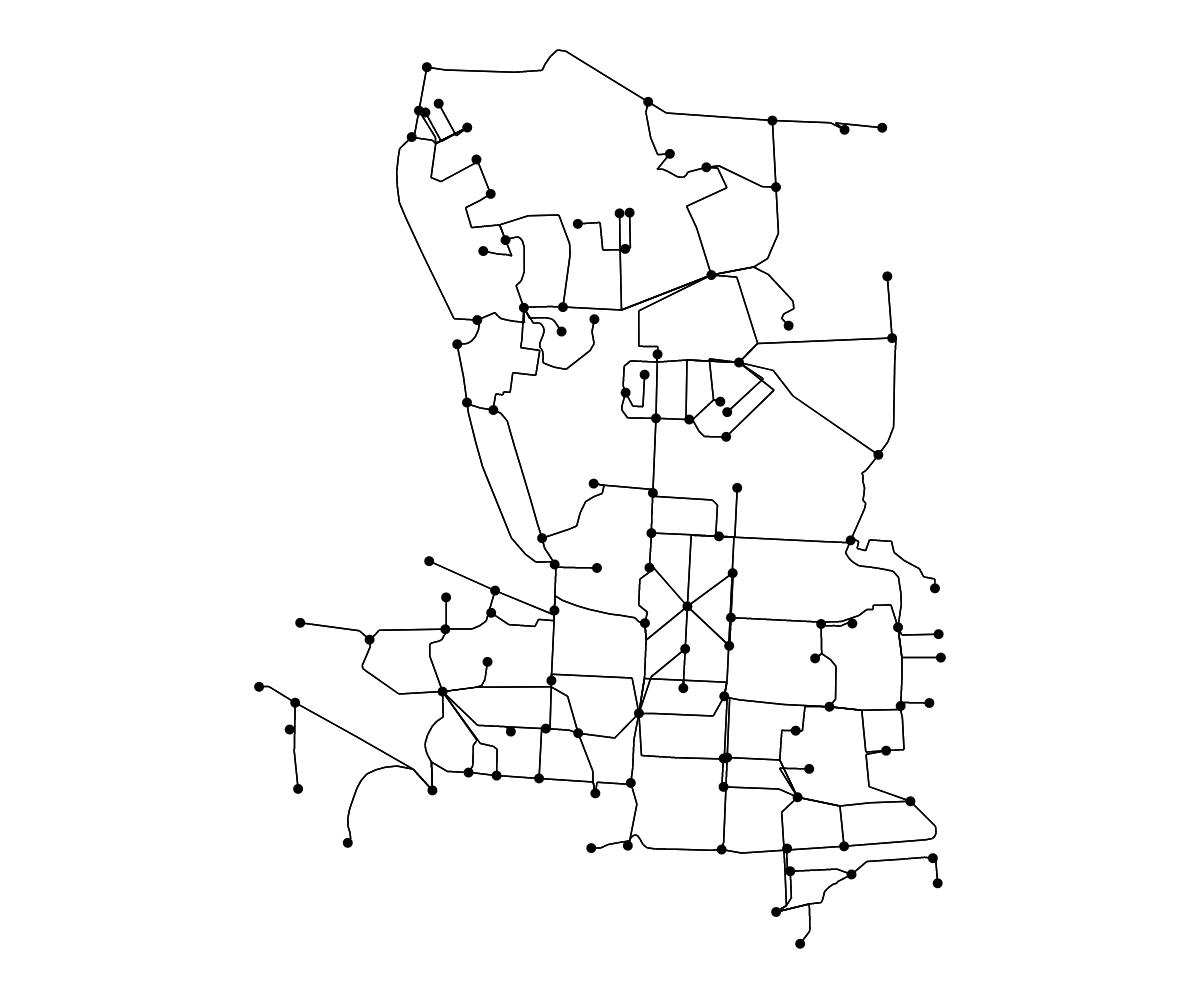}}
    \subfloat[Campus C]{\includegraphics[width=0.32\textwidth]{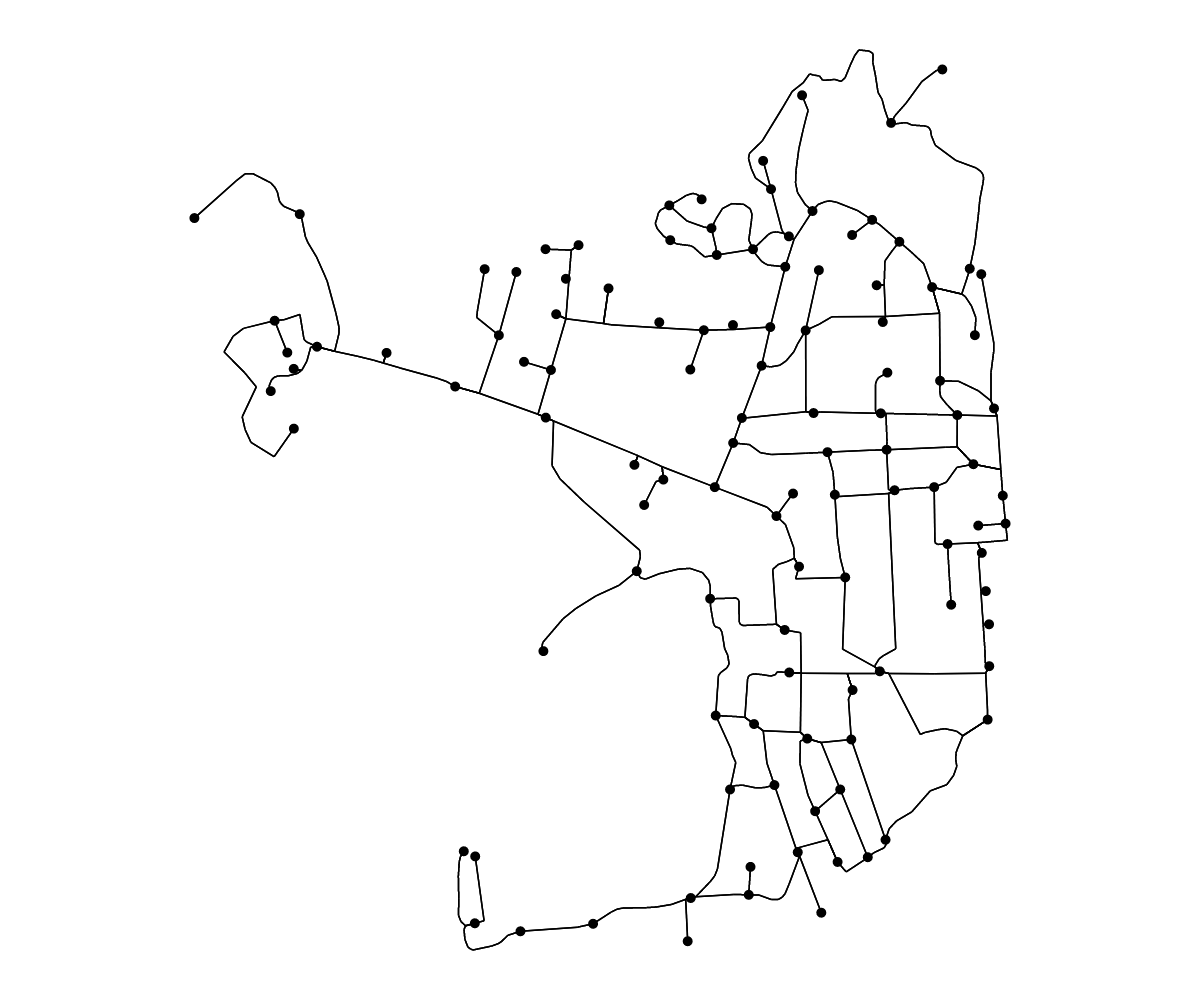}}\\
    \subfloat[Campus D]{\includegraphics[width=0.32\textwidth]{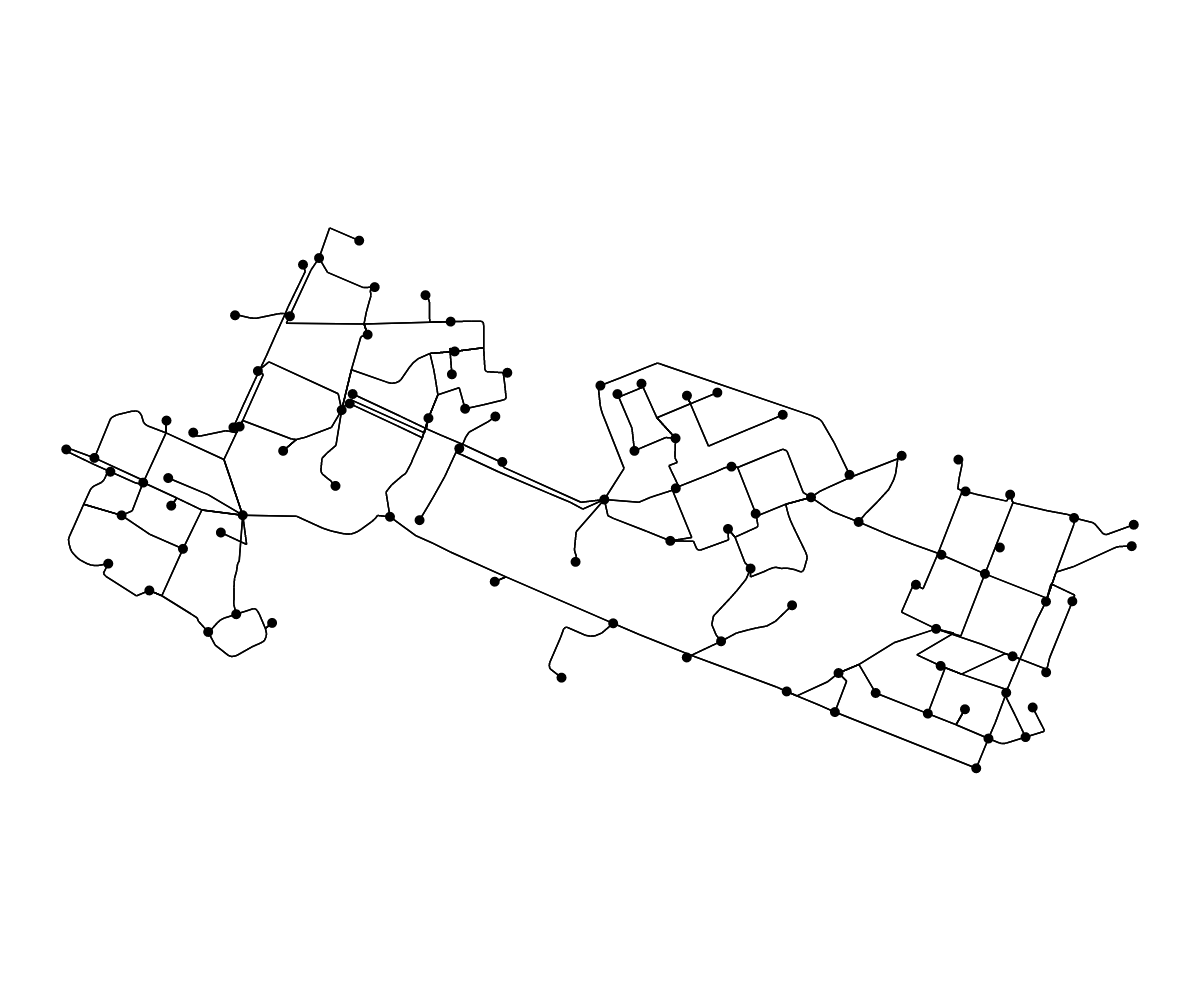}}
    \subfloat[Campus E]{\includegraphics[width=0.32\textwidth]{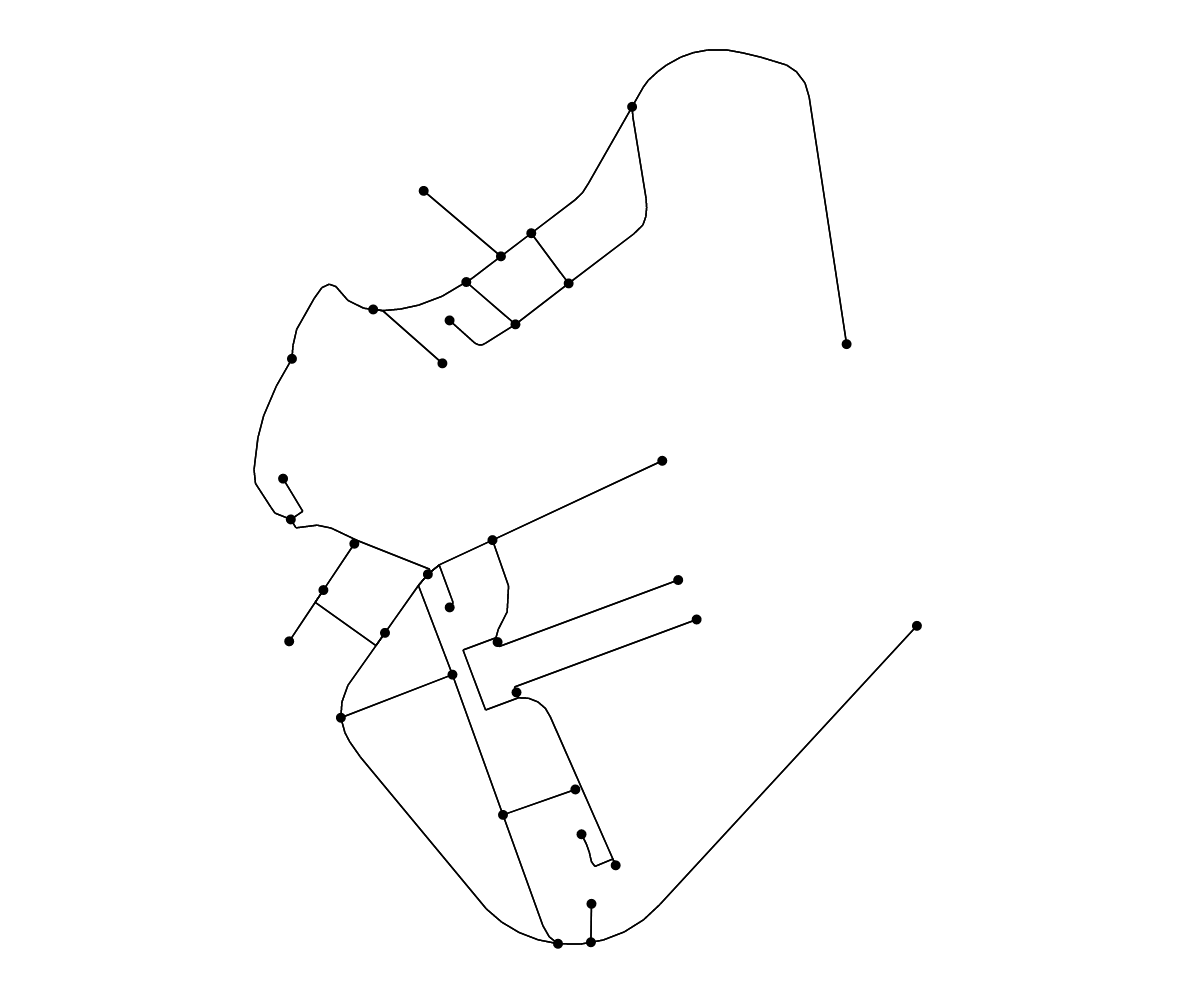}}
    \subfloat[Campus F]{\includegraphics[width=0.32\textwidth]{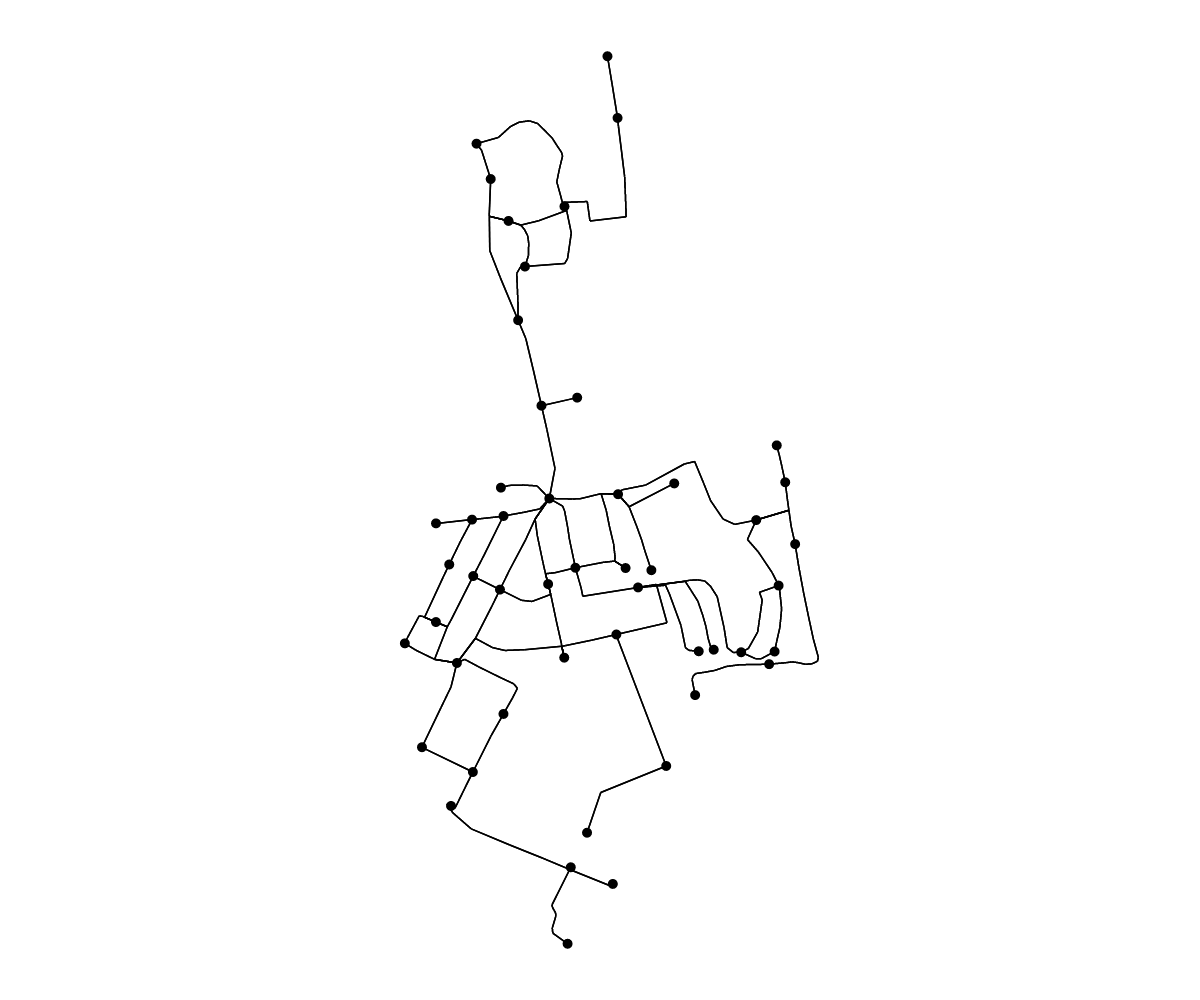}}
    }}
    \caption{Different real-world campus environments used for simulations.}\label{fig_campus}
\end{figure}

We use six real-world campus environments (Figure~\ref{fig_campus}). Campuses A, C, and E have long boundary walls, while the other layouts are more compact. Campus F has a small satellite portion connected via a single road segment separate from the main campus. 

\begin{table}[ht]
    \centering
    \begin{tabular}{|c|c|c|c|c|c|}
    \hline
    Layout & Type & $\epsilon(1)$ & $\vert \mathcal{V} \vert$ & $\vert \mathcal{E} \vert$ & $d_{avg}$\\
    \hline
    Campus A & City Block & 0.197 & 87 & 232 & 2.667 \\
    Campus B & University & 0.178 & 116 & 336 & 2.897\\
    Campus C & University & 0.192 & 113 & 288 & 2.549\\
    Campus D & University & 0.171 & 97 & 266 & 2.742\\
    Campus E & Military Base & 0.189 & 36 & 80 & 2.222\\
    Campus F & University & 0.185 & 50 & 128 & 2.56\\
    \hline
    \end{tabular}
    \caption{Attributes of the Campus Environments.}
    \label{tab_layout_details}
\end{table}

Table~\ref{tab_layout_details} details each environment's attributes. We run simulations with $5$ to $10$ agents patrolling each campus. Each run is of 432,000 seconds (5 days in simulation). The velocity of the agents is $10 m/s$. The value of $\epsilon(1)$ is set to ${1}/{\underline{w}}$. We determine five different recurrent patrol strategies for each input scenario with discretization constants $D = \{1, 2, 3, 4, 5\}$. The evaluation metrics are as follows-
\begin{itemize}
    \item Graph Maximum Idleness, \[GMI = \max_{t \geq 0, v_k \in \mathcal{V}} I_k(t)\]
    \item Graph Average Idleness, \[GAI = \max_{t \geq 0} \frac{\sum_{v_k \in \mathcal{V}} I_k(t)}{\vert \mathcal{V} \vert}\]
\end{itemize}
\begin{figure}[ht]
    \centering
    \fbox{\subfloat[]{\includegraphics[width=\textwidth]{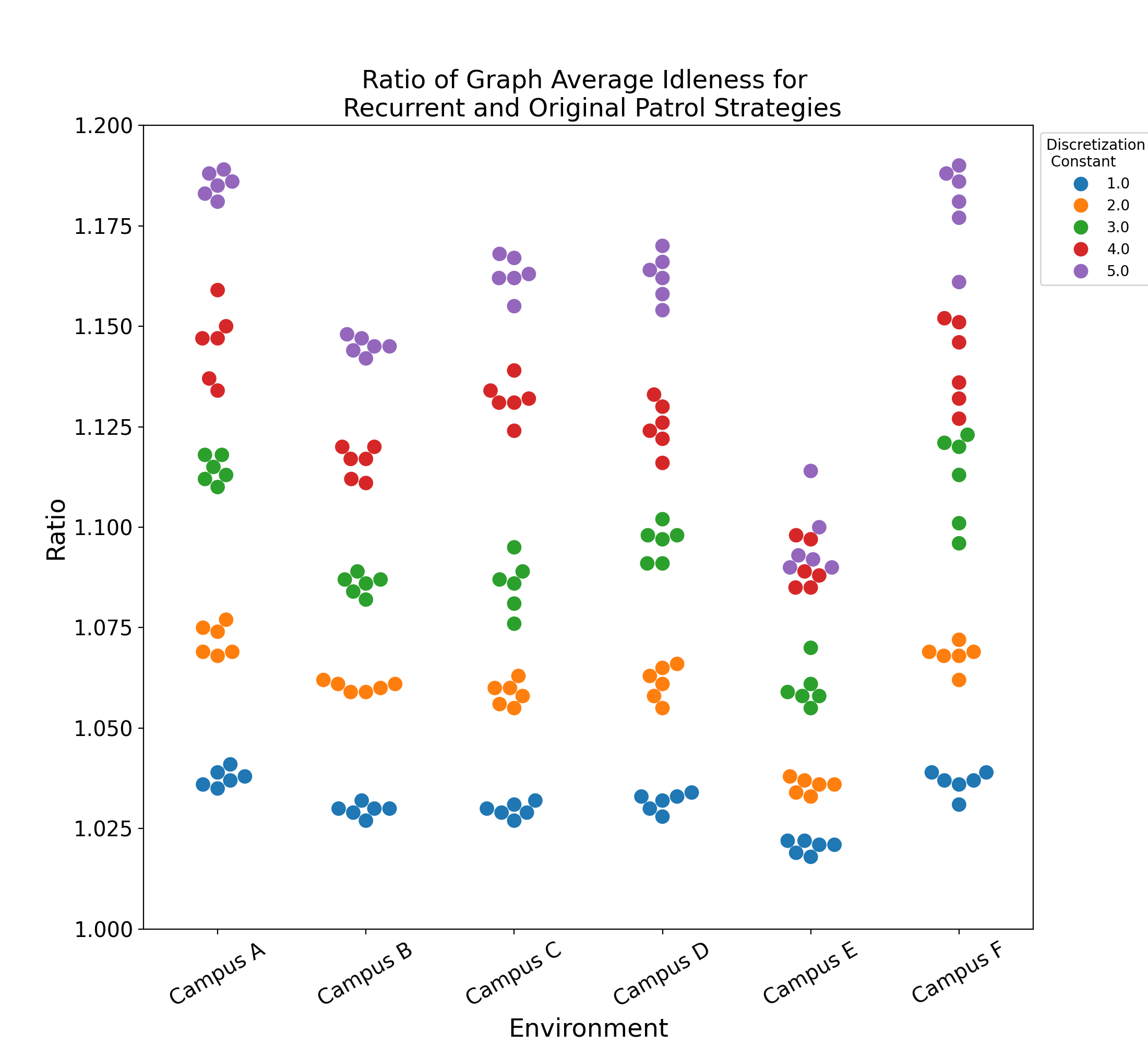}}} \\
    \fbox{\subfloat[]{\includegraphics[width=\textwidth]{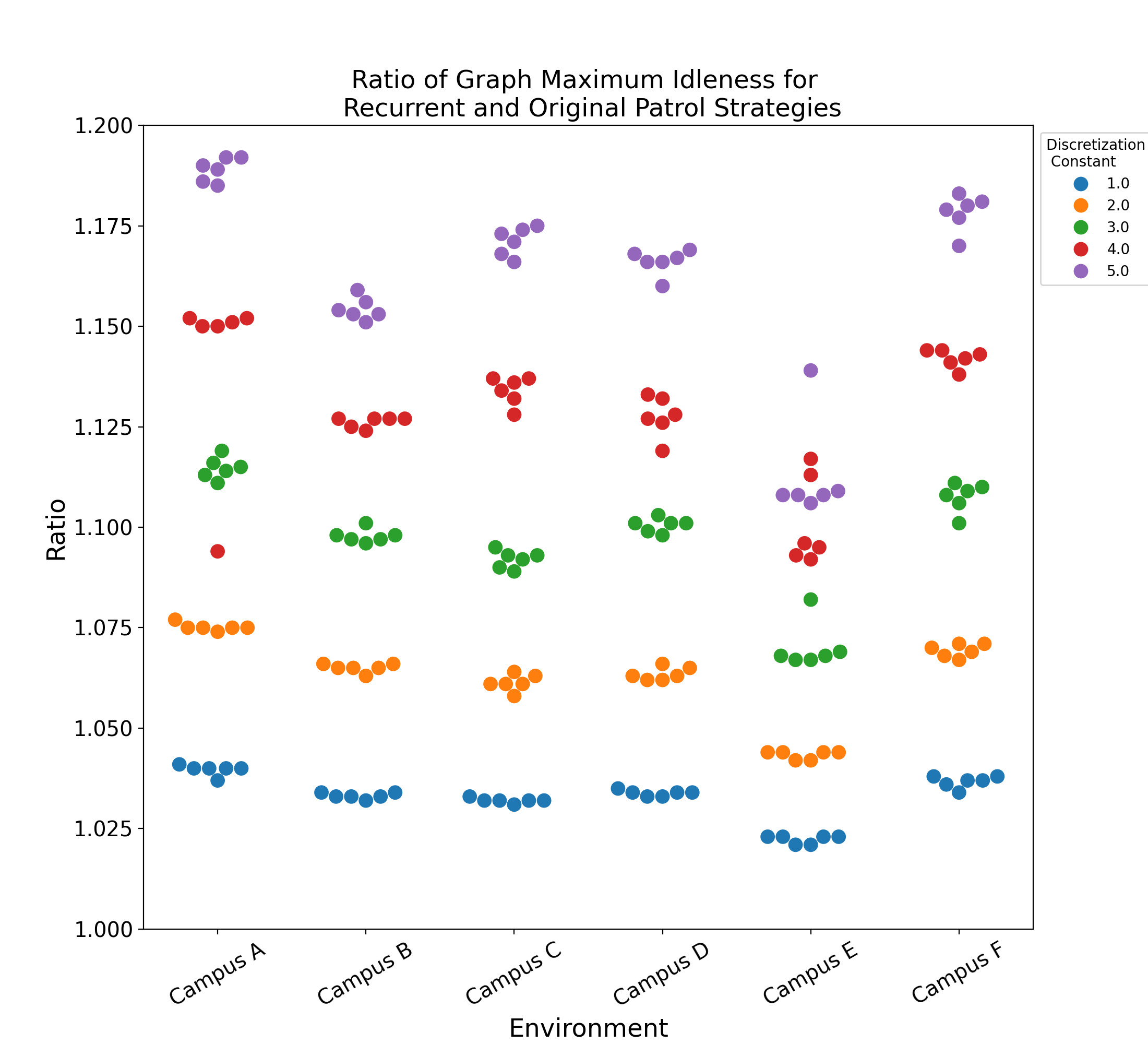}}}
    \caption{Ratios of (a) Graph Average Idleness and (b) Graph Maximum Idleness values under recurrent patrol strategies for various $D$ values to that under original patrol strategy.}\label{fig_gi}
\end{figure}
\begin{figure}[ht]
    \centering
    \fbox{\includegraphics[width=\textwidth]{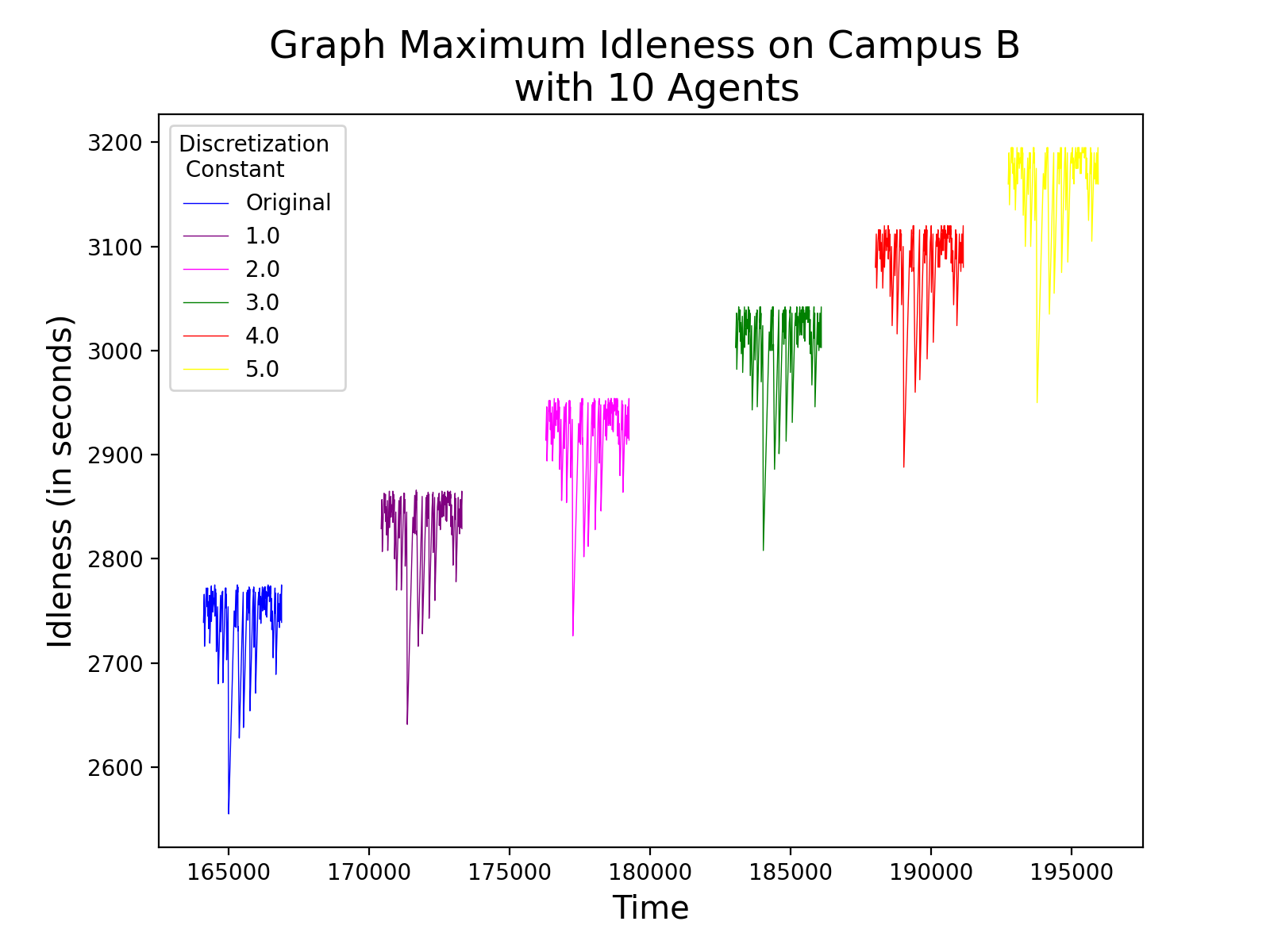}}
    \caption{Graph Maximum Idleness values over time of the recurring portion under original and discrete patrol strategies.}
    \label{fig_mi}
\end{figure}

Figure~\ref{fig_gi} shows the ratios of Graph Average Idleness (GAI) values and Graph Maximum Idleness (GMI) values for various discretization constant $D$, that is, $Ratio = J_\pi(D)/J_\pi$ where $J_\pi(D)$ is the cost (either GAI or GMI) of recurrent patrol strategy $\pi^R$ with discretization constant $D$ and $J_\pi$ is the corresponding cost under the original Greedy-Random patrol strategy $\pi$. Suppose the recurrent patrol strategy is determined by repeating the discrete patrol strategy segment between departure instances $p$ and $q$; the cost $J_\pi$ is calculated over the same segment in $\pi$. For example, in Figure~\ref{fig_mi}, we plot the Graph Maximum Idleness values over time for patrol simulation on Campus B with ten agents. Each blob of data corresponding to different $D$ values is for the same segment from $p$ to $q$ under both original $\pi$ and discrete $\pi^D$ patrol strategies. As the departure times are delayed under $\pi^D$s, we see the offsets for various $D$ values. Since the ratios are below $(1 + \epsilon(D))$ under all simulations, the recurrent patrol strategy obtained using Algorithms~\ref{alg_disc} and~\ref{alg_rec} validates Theorem~\ref{thm_rps}.

Tables~\ref{tab_gai} and~\ref{tab_gmi} contain the Graph Average Idleness (GAI) values and Graph Maximum Idleness (GMI) values for different simulations. In every case, the GAI/GMI value obtained by the recurrent patrol strategy for various values of $D$ is less than $(1 + \epsilon(D))$ times that of the original patrol strategy. For example, the Graph Average Idleness for five agents patrolling on Campus A under Greedy-Random Patrol (Algorithm~\ref{alg_grps}) is $GAI = 949$ whereas the recurrent patrol strategy generated with discretization constant $D = 3$ results in $GAI(3) = 1061$ which is less than $(1 + \epsilon(D)) \times GAI = 1.591 * 949 = 1510$. 

\clearpage
\onecolumn
\begin{longtable}{|c|c|c c|c c|c c|c c|c c|}
\caption{Graph Average Idleness values for all simulations with different discretization constant $D$ values.\\ (1) $J_\pi(D)$- Cost of recurrent patrol strategy $\pi^R_D$\\ (2) $\bar{J}_\pi(D) = (1 + \epsilon(D))J_\pi$- Upper Bound on $J_\pi(D)$ as per Theorem~\ref{thm_rps}}
\label{tab_gai}\\
\hline 
\multirow{3}{*}{\centering $\vert \mathcal{A} \vert$} & \multirow{3}{1.5cm}{\centering Original, $J_\pi$} & \multicolumn{2}{c|}{$D = 1$} & \multicolumn{2}{c|}{$D = 2$} & \multicolumn{2}{c|}{$D = 3$} & \multicolumn{2}{c|}{$D = 4$} & \multicolumn{2}{c|}{$D = 5$} \\  
\cline{3-12}
& & \multirow{2}{*}{$J_{\pi}(1)$} & \multirow{2}{*}{$\bar{J}_\pi(1)$} & \multirow{2}{*}{$J_{\pi}(2)$} & \multirow{2}{*}{$\bar{J}_\pi(2)$} & \multirow{2}{*}{$J_{\pi}(3)$} & \multirow{2}{*}{$\bar{J}_\pi(3)$} & \multirow{2}{*}{$J_{\pi}(4)$} & \multirow{2}{*}{$\bar{J}_\pi(4)$} & \multirow{2}{*}{$J_{\pi}(5)$} & \multirow{2}{*}{$\bar{J}_\pi(5)$} \\
& & & & & & & & & & & \\
\hline 
\endfirsthead

\multicolumn{12}{c}
{{\bfseries \tablename\ \thetable{} -- continued from previous page}} \\
\hline 
\multirow{3}{*}{\centering $\vert \mathcal{A} \vert$} & \multirow{3}{1.5cm}{\centering Original, $J_\pi$} & \multicolumn{2}{c|}{$D = 1$} & \multicolumn{2}{c|}{$D = 2$} & \multicolumn{2}{c|}{$D = 3$} & \multicolumn{2}{c|}{$D = 4$} & \multicolumn{2}{c|}{$D = 5$} \\  
\cline{3-12}
& & \multirow{2}{*}{$J_{\pi}(1)$} & \multirow{2}{*}{$\bar{J}_\pi(1)$} & \multirow{2}{*}{$J_{\pi}(2)$} & \multirow{2}{*}{$\bar{J}_\pi(2)$} & \multirow{2}{*}{$J_{\pi}(3)$} & \multirow{2}{*}{$\bar{J}_\pi(3)$} & \multirow{2}{*}{$J_{\pi}(4)$} & \multirow{2}{*}{$\bar{J}_\pi(4)$} & \multirow{2}{*}{$J_{\pi}(5)$} & \multirow{2}{*}{$\bar{J}_\pi(5)$} \\
& & & & & & & & & & & \\
\hline 
\endhead

\hline \multicolumn{12}{r}{{Continued on next page}} \\ \hline
\endfoot

\hline \hline
\endlastfoot

\multicolumn{12}{c}{Campus A, $\epsilon(1) = 0.197$} \\
\hline
5 & 949 & 986 & 1136 & 1020 & 1323 & 1061 & 1510 & 1092 & 1697 & 1128 & 1884 \\
6 & 950 & 989 & 1137 & 1023 & 1324 & 1062 & 1511 & 1101 & 1699 & 1130 & 1886 \\
7 & 950 & 986 & 1137 & 1016 & 1324 & 1058 & 1511 & 1090 & 1699 & 1122 & 1886 \\
8 & 1006 & 1042 & 1204 & 1075 & 1402 & 1119 & 1601 & 1141 & 1799 & 1192 & 1997 \\
9 & 963 & 998 & 1153 & 1029 & 1342 & 1074 & 1532 & 1095 & 1722 & 1140 & 1912 \\
10 & 1119 & 1162 & 1339 & 1203 & 1560 & 1242 & 1780 & 1284 & 2001 & 1327 & 2221 \\
\hline
\multicolumn{12}{c}{Campus B, $\epsilon(1) = 0.197$} \\
\hline
5 & 1457 & 1501 & 1716 & 1546 & 1976 & 1579 & 2235 & 1632 & 2494 & 1673 & 2754 \\
6 & 1484 & 1531 & 1748 & 1576 & 2012 & 1613 & 2276 & 1663 & 2541 & 1699 & 2805 \\
7 & 1498 & 1543 & 1765 & 1589 & 2031 & 1627 & 2298 & 1672 & 2565 & 1718 & 2831 \\
8 & 1474 & 1514 & 1736 & 1561 & 1999 & 1595 & 2261 & 1647 & 2523 & 1686 & 2786 \\
9 & 1490 & 1535 & 1755 & 1579 & 2020 & 1622 & 2286 & 1658 & 2551 & 1707 & 2816 \\
10 & 1520 & 1564 & 1791 & 1609 & 2061 & 1652 & 2332 & 1689 & 2602 & 1736 & 2873 \\
\hline
\multicolumn{12}{c}{Campus C, $\epsilon(1) = 0.197$} \\
\hline
5 & 1454 & 1500 & 1733 & 1545 & 2012 & 1592 & 2292 & 1656 & 2571 & 1698 & 2850 \\
6 & 1290 & 1329 & 1538 & 1367 & 1785 & 1401 & 2033 & 1462 & 2281 & 1504 & 2528 \\
7 & 1632 & 1680 & 1945 & 1727 & 2259 & 1756 & 2572 & 1849 & 2885 & 1897 & 3199 \\
8 & 1325 & 1361 & 1579 & 1398 & 1834 & 1432 & 2088 & 1489 & 2343 & 1541 & 2597 \\
9 & 1387 & 1427 & 1653 & 1466 & 1920 & 1511 & 2186 & 1570 & 2452 & 1603 & 2719 \\
10 & 1338 & 1379 & 1595 & 1418 & 1852 & 1455 & 2109 & 1514 & 2366 & 1556 & 2622 \\
\hline 
\multicolumn{12}{c}{Campus D, $\epsilon(1) = 0.197$} \\
\hline
5 & 1334 & 1380 & 1562 & 1423 & 1790 & 1470 & 2018 & 1512 & 2246 & 1561 & 2475 \\
6 & 1261 & 1297 & 1477 & 1330 & 1692 & 1376 & 1908 & 1408 & 2124 & 1460 & 2339 \\
7 & 1264 & 1302 & 1480 & 1337 & 1696 & 1379 & 1912 & 1418 & 2129 & 1458 & 2345 \\
8 & 1343 & 1385 & 1573 & 1425 & 1802 & 1472 & 2032 & 1509 & 2262 & 1562 & 2491 \\
9 & 1343 & 1388 & 1573 & 1427 & 1802 & 1475 & 2032 & 1512 & 2262 & 1566 & 2491 \\
10 & 1362 & 1406 & 1595 & 1451 & 1828 & 1495 & 2061 & 1539 & 2294 & 1583 & 2527 \\
\hline 
\multicolumn{12}{c}{Campus E, $\epsilon(1) = 0.197$} \\
\hline
5 & 637 & 651 & 757 & 662 & 878 & 677 & 998 & 694 & 1119 & 701 & 1239 \\
6 & 641 & 653 & 762 & 663 & 883 & 677 & 1004 & 552 & 1126 & 560 & 1247 \\
7 & 643 & 656 & 765 & 666 & 886 & 681 & 1008 & 553 & 1129 & 701 & 1251 \\
8 & 646 & 660 & 768 & 670 & 890 & 685 & 1012 & 703 & 1134 & 706 & 1256 \\
9 & 639 & 653 & 760 & 663 & 881 & 534 & 1001 & 693 & 1122 & 698 & 1243 \\
10 & 645 & 657 & 767 & 667 & 889 & 682 & 1011 & 700 & 1133 & 703 & 1255 \\
\hline
\multicolumn{12}{c}{Campus F, $\epsilon(1) = 0.197$} \\
\hline
5 & 696 & 722 & 825 & 744 & 954 & 780 & 1082 & 798 & 1211 & 827 & 1340 \\
6 & 612 & 631 & 725 & 650 & 838 & 671 & 952 & 690 & 1065 & 711 & 1178 \\
7 & 586 & 607 & 694 & 626 & 803 & 645 & 911 & 663 & 1020 & 690 & 1128 \\
8 & 689 & 716 & 816 & 737 & 944 & 773 & 1071 & 794 & 1199 & 818 & 1326 \\
9 & 663 & 689 & 786 & 710 & 908 & 745 & 1031 & 763 & 1154 & 789 & 1276 \\
10 & 618 & 641 & 732 & 661 & 847 & 688 & 961 & 702 & 1075 & 730 & 1190 \\
\hline
\end{longtable}

\begin{longtable}{|c|c|c c|c c|c c|c c|c c|}
\caption{Graph Maximum Idleness values for all simulations with different discretization constant $D$ values.\\ (1) $J_\pi(D)$- Cost of recurrent patrol strategy $\pi^R_D$\\ (2) $\bar{J}_\pi(D) = (1 + \epsilon(D))J_\pi$- Upper Bound on $J_\pi(D)$ as per Theorem~\ref{thm_rps}}
\label{tab_gmi} \\
\hline 
\multirow{3}{*}{\centering $\vert \mathcal{A} \vert$} & \multirow{3}{1.5cm}{\centering Original, $J_\pi$} & \multicolumn{2}{c|}{$D = 1$} & \multicolumn{2}{c|}{$D = 2$} & \multicolumn{2}{c|}{$D = 3$} & \multicolumn{2}{c|}{$D = 4$} & \multicolumn{2}{c|}{$D = 5$} \\  
\cline{3-12}
& & \multirow{2}{*}{$J_{\pi}(1)$} & \multirow{2}{*}{$\bar{J}_\pi(1)$} & \multirow{2}{*}{$J_{\pi}(2)$} & \multirow{2}{*}{$\bar{J}_\pi(2)$} & \multirow{2}{*}{$J_{\pi}(3)$} & \multirow{2}{*}{$\bar{J}_\pi(3)$} & \multirow{2}{*}{$J_{\pi}(4)$} & \multirow{2}{*}{$\bar{J}_\pi(4)$} & \multirow{2}{*}{$J_{\pi}(5)$} & \multirow{2}{*}{$\bar{J}_\pi(5)$} \\
& & & & & & & & & & & \\
\hline 
\endfirsthead

\multicolumn{12}{c}
{{\bfseries \tablename\ \thetable{} -- continued from previous page}} \\
\hline 
\multirow{3}{*}{\centering $\vert \mathcal{A} \vert$} & \multirow{3}{1.5cm}{\centering Original, $J_\pi$} & \multicolumn{2}{c|}{$D = 1$} & \multicolumn{2}{c|}{$D = 2$} & \multicolumn{2}{c|}{$D = 3$} & \multicolumn{2}{c|}{$D = 4$} & \multicolumn{2}{c|}{$D = 5$} \\  
\cline{3-12}
& & \multirow{2}{*}{$J_{\pi}(1)$} & \multirow{2}{*}{$\bar{J}_\pi(1)$} & \multirow{2}{*}{$J_{\pi}(2)$} & \multirow{2}{*}{$\bar{J}_\pi(2)$} & \multirow{2}{*}{$J_{\pi}(3)$} & \multirow{2}{*}{$\bar{J}_\pi(3)$} & \multirow{2}{*}{$J_{\pi}(4)$} & \multirow{2}{*}{$\bar{J}_\pi(4)$} & \multirow{2}{*}{$J_{\pi}(5)$} & \multirow{2}{*}{$\bar{J}_\pi(5)$} \\
& & & & & & & & & & & \\
\hline 
\endhead

\hline \multicolumn{12}{r}{{Continued on next page}} \\ \hline
\endfoot

\hline \hline
\endlastfoot

\multicolumn{12}{c}{Campus A, $\epsilon(1) = 0.197$} \\
\hline
5 & 1900 & 1977 & 2274 & 2046 & 2649 & 2121 & 3023 & 2188 & 3397 & 2260 & 3772 \\
6 & 1872 & 1941 & 2240 & 2010 & 2610 & 2079 & 2978 & 2152 & 3347 & 2220 & 3716 \\
7 & 1964 & 1940 & 2350 & 2006 & 2738 & 2076 & 3125 & 2148 & 3512 & 2220 & 3899 \\
8 & 1837 & 1910 & 2198 & 1974 & 2561 & 2055 & 2923 & 2116 & 3285 & 2190 & 3646 \\
9 & 1913 & 1989 & 2289 & 2056 & 2667 & 2133 & 3044 & 2200 & 3420 & 2280 & 3797 \\
10 & 3130 & 3255 & 3746 & 3366 & 4363 & 3486 & 4980 & 3604 & 5596 & 3710 & 6213 \\
\hline
\multicolumn{12}{c}{Campus B, $\epsilon(1) = 0.178$} \\
\hline
5 & 2637 & 2726 & 3106 & 2812 & 3576 & 2895 & 4045 & 2972 & 4515 & 3055 & 4984 \\
6 & 2772 & 2867 & 3265 & 2952 & 3759 & 3051 & 4252 & 3124 & 4746 & 3205 & 5239 \\
7 & 2710 & 2800 & 3192 & 2886 & 3675 & 2976 & 4157 & 3048 & 4640 & 3125 & 5122 \\
8 & 2641 & 2725 & 3111 & 2808 & 3581 & 2898 & 4051 & 2976 & 4521 & 3045 & 4991 \\
9 & 2748 & 2838 & 3237 & 2928 & 3726 & 3015 & 4215 & 3096 & 4705 & 3170 & 5194 \\
10 & 2775 & 2866 & 3268 & 2954 & 3763 & 3042 & 4257 & 3120 & 4751 & 3195 & 5245 \\
\hline
\multicolumn{12}{c}{Campus C, $\epsilon(1) = 0.192$} \\
\hline
5 & 3708 & 3810 & 4419 & 3924 & 5132 & 4044 & 5844 & 4200 & 6556 & 4330 & 7268 \\
6 & 2444 & 2524 & 2913 & 2600 & 3382 & 2670 & 3852 & 2780 & 4321 & 2870 & 4790 \\
7 & 3952 & 4074 & 4710 & 4192 & 5470 & 4302 & 6228 & 4472 & 6987 & 4610 & 7746 \\
8 & 2468 & 2546 & 2941 & 2612 & 3416 & 2691 & 3890 & 2784 & 4363 & 2890 & 4837 \\
9 & 2506 & 2585 & 2987 & 2660 & 3468 & 2739 & 3949 & 2848 & 4431 & 2940 & 4912 \\
10 & 2490 & 2569 & 2968 & 2642 & 3446 & 2721 & 3924 & 2824 & 4402 & 2925 & 4880 \\
\hline
\multicolumn{12}{c}{Campus D, $\epsilon(1) = 0.171$} \\
\hline
5 & 4862 & 5032 & 5693 & 5182 & 6525 & 5364 & 7356 & 5508 & 8188 & 5685 & 9019 \\
6 & 2427 & 2506 & 2842 & 2578 & 3257 & 2664 & 3672 & 2732 & 4087 & 2830 & 4502 \\
7 & 4525 & 4630 & 5298 & 4760 & 6073 & 4923 & 6846 & 5064 & 7620 & 5200 & 8394 \\
8 & 4874 & 5041 & 5707 & 5182 & 6541 & 5364 & 7374 & 5496 & 8208 & 5685 & 9041 \\
9 & 4873 & 5041 & 5706 & 5182 & 6540 & 5364 & 7373 & 5492 & 8206 & 5685 & 9039 \\
10 & 4920 & 5089 & 5761 & 5242 & 6603 & 5418 & 7444 & 5568 & 8285 & 5745 & 9127 \\
\hline
\multicolumn{12}{c}{Campus E, $\epsilon(1) = 0.189$} \\
\hline  
5 & 1286 & 1315 & 1529 & 1342 & 1772 & 1374 & 2015 & 1408 & 2258 & 1425 & 2501 \\
6 & 1288 & 1315 & 1531 & 1342 & 1775 & 1374 & 2018 & 1276 & 2262 & 1305 & 2505 \\
7 & 1286 & 1315 & 1529 & 1342 & 1772 & 1374 & 2015 & 1276 & 2258 & 1425 & 2501 \\
8 & 1285 & 1315 & 1527 & 1342 & 1771 & 1374 & 2014 & 1408 & 2256 & 1425 & 2499 \\
9 & 1286 & 1315 & 1529 & 1342 & 1772 & 1239 & 2015 & 1404 & 2258 & 1425 & 2501 \\
10 & 1288 & 1315 & 1531 & 1342 & 1775 & 1374 & 2018 & 1408 & 2262 & 1425 & 2505 \\
\hline
\multicolumn{12}{c}{Campus F, $\epsilon(1) = 0.185$} \\
\hline
5 & 1765 & 1830 & 2091 & 1886 & 2418 & 1959 & 2745 & 2016 & 3071 & 2085 & 3398 \\
6 & 1188 & 1228 & 1407 & 1268 & 1628 & 1308 & 1847 & 1352 & 2067 & 1390 & 2287 \\
7 & 1137 & 1180 & 1347 & 1218 & 1558 & 1257 & 1768 & 1300 & 1978 & 1340 & 2189 \\
8 & 1755 & 1820 & 2079 & 1874 & 2404 & 1947 & 2729 & 2008 & 3054 & 2065 & 3378 \\
9 & 1763 & 1830 & 2089 & 1886 & 2415 & 1959 & 2741 & 2016 & 3068 & 2085 & 3394 \\
10 & 1199 & 1242 & 1420 & 1284 & 1643 & 1329 & 1864 & 1368 & 2086 & 1415 & 2308 \\
\hline
\end{longtable}
\clearpage
\twocolumn
\section{Conclusion} \label{sec_con}
Through this work, we introduce a novel problem formulation, `General Patrol Problem' (Problem~\ref{prob_gpp}) for Multi-Robot Patrolling Problems that generalizes over many existing problems studied in the literature. Our significant contributions are as follows. We establish that an $\epsilon$-approximate recurrent patrol strategy exists for every patrol strategy that solves the General Patrol Problem. The $\epsilon$ factor is a function of a parameter discretization constant $D$ that can be arbitrarily small. We also prove an $\epsilon$-optimal solution to Problem~\ref{prob_gpp} exists. Through this, we establish a key result that aids in obtaining near-optimal solutions to many problem formulations. 
We improve the approximation factor to almost $D/\underline{w}$ where $D$ is a parameter (discretization constant), and $\underline{w}$ is the minimum edge weight of the graph. Hence, $\epsilon$ is independent of the number of patrol agents and the graph size and can be arbitrarily small by selecting $D$ appropriately.

We also describe a systematic approach to derive an $\epsilon$-approximate recurrent patrol strategy for any given patrol strategy as Algorithms~\ref{alg_disc} and~\ref{alg_rec}. We validate our results through extensive simulations.
We use Greedy Random Patrol presented in Algorithm~\ref{alg_grps} to generate instances of realistic patrol deployment on different real-life campus environments. The results establish that the $\epsilon$ factor holds under all cases.

Through this work, we have commented on the nature of optimal solutions (their existence and class) to various formulations of the Multi-Robot Patrol Problem. We can re-analyze the approximation factors of algorithms in the literature that generate recurrent patrol strategies. While the class of recurrent patrol strategies is expansive, it is finite and much smaller than that of all feasible patrol strategies. We invite the community to focus on this class of solutions and help narrow it down further. 
\bibliographystyle{abbrv}
\bibliography{refer}
\end{document}